\documentclass[11pt]{article}
\usepackage[pdftex,paper=a4paper,left=3cm,right=3cm,top=4cm,bottom=4cm]{geometry}
\usepackage{amsthm, amsmath, amssymb} 
\usepackage{mathrsfs}         
\usepackage{graphicx}
\usepackage{url} 
\usepackage{algorithmic}
\usepackage{algorithm}
\usepackage[T1]{fontenc}
\usepackage{lmodern}

\newtheorem{theorem}{Theorem}
\newtheorem{lemma}[theorem]{Lemma}

\newtheorem{definition}[theorem]{Definition}
\newtheorem{corollary}[theorem]{Corollary}

\newtheorem{example}[theorem]{Example}
\newtheorem{property}[theorem]{Property}
 
\newcommand{\RR}{\mathbb{R}}
\newcommand{\F}{\mathcal{F}}
\newcommand{\fmax}{f_{\max}}
\newcommand{\e}{\varepsilon}
\renewcommand{\sp}{\ell}
\newcommand{\COST}{\mathscr{C}}
\newcommand{\event}{\mathsf{E}}

\newcommand{\G}{\overleftrightarrow{G}}

\newcommand{\DOT}{\,.}
\newcommand{\COMMA}{\,,}
\newcommand{\WHERE}{\,\colon\,}

\newcommand{\tIF}{\text{if}}
\newcommand{\tAND}{\text{and}}
\DeclareMathOperator{\Reconstruct}{\mathsf{Reconstruct}}
\DeclareMathOperator{\sgn}{sgn}

\newcommand{\Ex}[2][]{\operatorname{\textbf{E}}_{#1}\left[#2\right]}
\renewcommand{\Pr}[2][]{\operatorname{\textbf{Pr}}_{#1}\left[#2\right]}
\newcommand{\SET}[1]{\left\{#1\right\}}
\newcommand{\sSET}[1]{\{#1\}}

\newcommand{\floor}[1]{\left\lfloor#1\right\rfloor}
\newcommand{\PATH}[1][]{\stackrel{#1}{\rightsquigarrow}}
\newcommand{\spm}[1][]{\sp^{#1}_{-}}
\newcommand{\spp}[1][]{\sp^{#1}_{+}}
\newcommand{\REVERSE}[1]{\stackrel{\leftarrow}{#1}}

\newcommand{\GFX}[2][]{\includegraphics[#1]{#2.pdf}}

\newenvironment{fig}
{\begin{figure}[th]\begin{center}}
{\end{center}\end{figure}}

\newenvironment{GFXFIG}[2][]
{\begin{fig}\GFX[#1]{#2}}
{\end{fig}}

\begin{document}

\title{Smoothed Analysis of the Successive~Shortest~Path~Algorithm\thanks{This research was supported by ERC Starting Grant 306465 (BeyondWorstCase)
and NWO grant 613.001.023. The upper bound (Theorem 1) of this paper has been presented at the 24th ACM-SIAM Symp.\ on Discrete Algorithms (SODA 2013).}}

\newcommand*\samethanks[1][\value{footnote}]{\footnotemark[#1]}

\author{Tobias Brunsch\thanks{University of Bonn,
                Department of Computer Science,
                Germany.
                Email: {\tt \{brunsch,roeglin,roesner\}@cs.uni-bonn.de}}
 \and Kamiel Cornelissen\thanks{University of Twente,
              Department of Applied Mathematics,
              Enschede, The Netherlands.
              Email: {\tt \{k.cornelissen,b.manthey\}@utwente.nl}}
 \and Bodo Manthey\samethanks[3]
 \and Heiko R{\"o}glin\samethanks[2]
 \and Clemens R{\"o}sner\samethanks[2]
}

\date{}

\maketitle

\begin{abstract}
The minimum-cost flow problem is a classic problem in combinatorial optimization with various applications. Several pseudo-poly\-no\-mi\-al,
polynomial, and strongly polynomial algorithms have been developed in the past decades, and it seems that both the problem and the
algorithms are well understood. However, some of the algorithms' running times observed in empirical studies contrast the running times
obtained by worst-case analysis not only in the order of magnitude but also in the ranking when compared to each other. For example, the
Successive Shortest Path (SSP) algorithm, which has an exponential worst-case running time, seems to outperform the strongly polynomial
Minimum-Mean Cycle Canceling algorithm.

To explain this discrepancy, we study the SSP algorithm in the framework of smoothed analysis and
establish a bound of $O(mn\phi)$ for the number of iterations, which implies a smoothed running time of $O(mn\phi (m + n\log n))$, 
where~$n$ and~$m$ denote the number of nodes and edges, respectively, and~$\phi$ is a measure for the amount of random noise. 
This shows that worst-case instances for the SSP algorithm are not robust and unlikely to be encountered in practice. Furthermore, we prove a
smoothed lower bound of $\Omega(m \cdot \min \SET{ n, \phi } \cdot \phi)$ for the number of iterations of the SSP algorithm, showing that the
upper bound cannot be improved for $\phi = \Omega(n)$.
\end{abstract}

\section{Introduction}

Flow problems have gained a lot of attention in the second half of the twentieth century to model, for example, transportation and communication networks~\cite{DBLP:books/daglib/0069809,ForFul62}. Plenty of algorithms have been developed over the last fifty years. The first pseudo-polynomial algorithm for the minimum-cost flow problem was the Out-of-Kilter algorithm independently proposed by Minty~\cite{Min60} and by Fulkerson~\cite{Ful61}. The simplest pseudo-polynomial algorithms are the primal Cycle Canceling algorithm by Klein~\cite{Kle67} and the dual Successive Shortest Path (SSP) algorithm by Jewell~\cite{Jew62}, Iri~\cite{Iri60}, and Busacker and Gowen~\cite{BusGow60}. By introducing a scaling technique Edmonds and Karp~\cite{DBLP:journals/jacm/EdmondsK72} modified the SSP algorithm to obtain the Capacity Scaling algorithm, which was the first polynomial time algorithm for the minimum-cost flow problem.

The first strongly polynomial algorithms were given by Tardos~\cite{DBLP:journals/combinatorica/Tardos85} and by Orlin~\cite{Orl84}. Later,
Goldberg and Tarjan~\cite{DBLP:journals/jacm/GoldbergT89} proposed a pivot rule for the Cycle Canceling algorithm to obtain the strongly
polynomial Minimum-Mean Cycle Canceling (MMCC) algorithm. The fastest known strongly polynomial algorithm up to now is the Enhanced Capacity
Scaling algorithm due to Orlin~\cite{Orl93} and has a running time of $O(m \log(n) (m + n\log n))$, where~$n$ and~$m$ denote the number of nodes and edges, respectively. For an extensive overview of
minimum-cost flow algorithms we suggest the paper of Goldberg and Tarjan~\cite{GolTar90}, the paper of Vygen~\cite{Vyg02}, and the book of
Ahuja, Magnanti, and Orlin~\cite{DBLP:books/daglib/0069809}.

Zadeh~\cite{Zad73} showed that the SSP algorithm has an exponential worst-case running time. Contrary to this,
the worst-case running times of the Capacity Scaling algorithm and the MMCC algorithm are $O(m (\log U) (m + n\log n))$ \cite{DBLP:journals/jacm/EdmondsK72} and $O(m^2n^2 \min \{ \log(nC), m \})$ \cite{DBLP:journals/algorithmica/RadzikG94}, respectively. Here,~$U$ denotes the maximum edge capacity and~$C$ denotes the maximum edge cost. In particular, the former is polynomial whereas the latter is even strongly polynomial.
However, the notions of pseudo-polynomial, polynomial, and strongly polynomial algorithms always refer to worst-case running times,
which do not always resemble the algorithms' behavior on real-life instances. Algorithms with large worst-case running times do not inevitably perform poorly in practice. An experimental study of Kov{\'a}cs~\cite{Kir14} indeed observes running time behaviors significantly deviating from what the worst-case running times indicate. The MMCC algorithm is completely outperformed by the SSP algorithm. The Capacity Scaling algorithm is the fastest of these three algorithms, but its running time seems to be in the same order of magnitude as the running time of the SSP algorithm.
In this article, we explain why the SSP algorithm comes off so well by applying the framework of smoothed analysis.

Smoothed analysis was introduced by Spielman and Teng~\cite{DBLP:journals/jacm/SpielmanT04} to explain why the simplex method is efficient in practice despite its exponential worst-case running time. In the original model, an adversary chooses an arbitrary instance which is subsequently slightly perturbed at random. In this way, pathological instances no longer dominate the analysis. Good smoothed bounds usually indicate good behavior in practice because in practice inputs are often subject to a small amount of random noise. For
instance, this random noise can stem from measurement errors, numerical imprecision, or rounding errors. It can also model influences that cannot be quantified exactly but for which there is no reason to believe that they are adversarial. Since its invention, smoothed analysis has been successfully applied in a variety of contexts. Two recent surveys~\cite{MR,SpielmanT09} summarize some of these results.

We follow a more general model of smoothed analysis due to Beier and V{\"o}cking~\cite{BeierV04}. In this model, the adversary is even
allowed to specify the probability distribution of the random noise. The power of the adversary is only limited by the \emph{smoothing
parameter}~$\phi$. In particular, in our input model the adversary does not fix the edge costs~$c_e \in [0, 1]$, but he
specifies for each edge~$e$ a probability density function~$f_e \colon [0, 1] \to [0, \phi]$ according to which the costs~$c_e$ are randomly
drawn independently of the other edge costs. If $\phi = 1$, then the adversary has no choice but to specify a uniform distribution on the
interval~$[0, 1]$ for each edge cost. In this case, our analysis becomes an average-case analysis. On the other hand, if~$\phi$ becomes
large, then the analysis approaches a worst-case analysis since the adversary can specify a small interval~$I_e$ of length~$1/\phi$ (which
contains the worst-case costs) for each edge~$e$ from which the costs~$c_e$ are drawn uniformly.

As in the worst-case analysis, the network graph, the edge capacities, and the balance values of the nodes are chosen adversarially. The edge capacities and the balance values of the nodes are even allowed to be real values.
We define the smoothed running time of an algorithm as the worst expected running time the adversary can achieve and we prove the following theorem.

\begin{theorem}
\label{maintheorem}
The SSP algorithm requires $O(mn\phi)$ augmentation steps in expectation and its smoothed running time is $O(mn\phi (m + n\log n))$.
\end{theorem}

If~$\phi$ is a constant -- which seems to be a reasonable assumption if it models, for example, measurement errors -- then the smoothed bound simplifies to $O(mn (m + n\log n))$. Hence, it is unlikely to encounter instances on which the SSP algorithm requires an exponential amount of time.

The following theorem, which we also prove in this article, states that the bound for the number of iterations of the SSP algorithm stated in Theorem~\ref{maintheorem} cannot be improved for $\phi = \Omega(n)$. 

\begin{theorem}
\label{theorem:lower bound}
For given positive integers $n$, $m \in \sSET{ n, \ldots, n^2 }$, and $\phi \le 2^{n}$ there exists a minimum-cost flow network with $O(n)$ nodes, $O(m)$ edges, and random edge costs with smoothing parameter~$\phi$ on which the SSP algorithm requires $\Omega(m \cdot \min \SET{ n, \phi } \cdot \phi)$ augmentation steps with probability~1.
\end{theorem}

The main technical section of this article is devoted to the proof of Theorem~\ref{maintheorem}
(Section~\ref{sec:analysis}). In Section~\ref{sec:lower bound} we derive the lower bound stated in Theorem~\ref{theorem:lower bound}.
At the end of this article (Section~\ref{sec:simplex}), we point out
some connections between SSP and its smoothed analysis
to the simplex method with the shadow vertex pivot rule, which has been used
by Spielman and Teng in their smoothed analysis~\cite{DBLP:journals/jacm/SpielmanT04}.

\subsection{The Minimum-Cost Flow Problem}

A \emph{flow network} is a simple directed graph~$G = (V, E)$ together with a \emph{capacity function}~$u \colon E \to \RR_{\geq 0}$. For convenience, we assume that there are no directed cycles of length two. In the minimum-cost flow problem there are an additional \emph{cost function}~$c \colon E \to [0, 1]$ and a \emph{balance function}~$b \colon V \to \RR$ indicating how much of a resource some node~$v$ requires ($b(v) < 0$) or offers ($b(v) > 0$). A \emph{feasible $b$-flow} for such an instance is a function~$f \colon E \to \RR_{\geq 0}$ that obeys the capacity constraints $0 \leq f_e \leq u_e$ for any edge~$e \in E$ and Kirchhoff's law adapted to the balance values, i.e.,
$
  b(v) + \sum_{e = (u, v) \in E} f_e = \sum_{e' = (v, w) \in E} f_{e'}
$
for all nodes $v \in V$. (Even though~$u$, $c$, and~$f$ are functions, we use the notation~$u_e$, $c_e$, and $f_e$ instead of~$u(e)$, $c(e)$, and $f(e)$ in this article.) If $\sum_{v \in V} b(v) \neq 0$, then there does not exist a feasible $b$-flow. We therefore always require $\sum_{v \in V} b(v) = 0$.
The cost of a feasible $b$-flow is defined as $c(f) = \sum_{e \in E} f_e \cdot c_e$. In the \emph{minimum-cost flow problem} the goal is to find the cheapest feasible $b$-flow, a so-called \emph{minimum-cost $b$-flow}, if one exists, and to output an error otherwise.

\subsection{The SSP Algorithm}
\label{sec:SSPAlg}

For a pair $e = (u, v)$, we denote by~$e^{-1}$ the pair $(v, u)$. Let~$G$ be a flow network, let~$c$ be a cost function, and let~$f$ be a flow. The \emph{residual network}~$G_f$ is the directed graph with vertex set~$V$, arc set $E' = E_\text{f} \cup E_\text{b}$, where
\[
  E_\text{f} = \bigl\{ e \WHERE e \in E \ \tAND\ f_e < u_e \bigr\}
\]
is the set of so-called \emph{forward arcs} and
\[
 E_\text{b} = \bigl\{ e^{-1} \WHERE e \in E \ \tAND\ f_e > 0 \bigr\}
\]
is the set of so-called \emph{backward arcs}, a capacity function $u' \colon E' \to \RR$, defined by
\[
  u'_e = \begin{cases}
    u_e - f_e & \tIF\ e \in E \COMMA \cr
    f_{e^{-1}} & \tIF\ e^{-1} \in E \COMMA
  \end{cases}
\]
and a cost function $c' \colon E' \to \RR$, defined by
\[
  c'_e = \begin{cases}
    c_e & \tIF\ e \in E \COMMA \cr
    -c_{e^{-1}} & \tIF\ e^{-1} \in E \DOT
  \end{cases}
\]
In practice, the simplest way to implement the SSP algorithm is to transform the instance to an equivalent instance with only one \emph{supply node} (a node with positive balance value) and one \emph{demand node} (a node with negative balance value). For this, we add two nodes~$s$ and~$t$ to the network which we call \emph{master source} and \emph{master sink}, edges $(s, v)$ for any supply node~$v$, and edges $(w, t)$ for any demand node~$w$. The capacities of these \emph{auxiliary edges} $(s, v)$ and $(w, t)$ are set to $b(v) > 0$ and $-b(w) > 0$, respectively. The costs of the auxiliary edges are set to~$0$. Now we set $b(s) = -b(t) = z$ where~$z$ is the sum of the capacities of the auxiliary edges incident with~$s$ (which is equal to the sum of the capacities of the auxiliary edges incident with~$t$ due to the assumption that~$\sum_{v \in V} b(v) = 0$). All other balance values are set to $0$.

This is a well-known transformation of an arbitrary minimum-cost flow instance into a minimum-cost flow instance with only a single source~$s$, a single sink~$t$, and~$b(v) = 0$ for all nodes~$v \in V\setminus\{s,t\}$. Nevertheless, we cannot assume without loss of generality that the flow network we study has only a single source and a single sink. The reason is that in the probabilistic input model introduced above it is not possible to insert auxiliary edges with costs~$0$ because the costs of each edge are chosen according to some density function that is bounded from above by~$\phi$. We have to consider the auxiliary edges with costs~$0$ explicitly and separately from the other edges in our analysis.

The SSP algorithm run on the transformed instance computes the minimum-cost $b$-flow for the original instance.
In the remainder of this article we use the term \emph{flow} to refer to a feasible $b$-flow for an arbitrary~$b$ with $b(s) = -b(t)$ and $b(v) = 0$ for $v \notin \SET{ s, t }$.
We will denote by $|f|$ the amount of flow shipped from $s$ to $t$ in flow $f$,
i.e.,
$|f| = \sum_{e = (s, v) \in E} f_e - \sum_{e = (v, s) \in E} f_e$.

The SSP algorithm for a minimum-cost flow network with a single source~$s$, a single sink~$t$, and with $b(s) = -b(t) = z > 0$ is given as Algorithm~\ref{algorithm:SSP}.

\begin{algorithm*}
  \caption{SSP for single-source-single-sink minimum-cost flow networks with $b(s) = -b(t) = z > 0$.}
  \label{algorithm:SSP}
  \begin{algorithmic}[1]
    \STATE start with the empty flow~$f_0 = 0$
    \FOR{$i = 1, 2, \ldots$}
      \STATE \textbf{if} $G_{f_{i-1}}$ does not contain a (directed) $s$-$t$ path \textbf{then} output that there does not exist a flow with value~$z$
      \STATE find a shortest $s$-$t$ path~$P_i$ in~$G_{f_{i-1}}$ with respect to the arc costs
      \STATE augment the flow as much as possible$^*$ along path~$P_i$ to obtain a new flow~$f_i$
      \STATE \textbf{if} $|f_i| = z$ \textbf{then} output~$f_i$
    \ENDFOR
  \end{algorithmic}
  \medskip

  \small $^*$ Since the value~$|f_i|$ of flow~$f_i$ must not exceed~$z$ and the flow~$f_i$ must obey all capacity constraints, the flow is increased by the minimum of $\min\{u_e - f_{i-1}(e) \mid e \in P_i \cap E\}$, $\min\{f_{i-1}(e) \mid e \in P_i~\wedge\stackrel{\leftarrow}{e} \in E\}$ and $z - |f_{i-1}|$.
\end{algorithm*}

\begin{theorem}
\label{thm:AllFlowsOpt}
In any round~$i$, flow~$f_i$ is a minimum-cost $b_i$-flow for the balance function~$b_i$ defined by
$
  b_i(s) = -b_i(t) = |f_i|
$
and
$
  b_i(v) = 0 \ \text{for~$v\notin\{s,t\}$}
$.
\end{theorem}

Theorem~\ref{thm:AllFlowsOpt} is due to Jewell~\cite{Jew62}, Iri~\cite{Iri60}, and Busacker and Gowen~\cite{BusGow60}. We refer to Korte and Vygen~\cite{Korte:2007:COT:1564997} for a proof.
As a consequence, no residual network~$G_{f_i}$ contains a directed cycle with negative total costs. Otherwise, we could augment along such a cycle to obtain a $b_i$-flow~$f'$ with smaller costs than~$f_i$. In particular, this implies that the shortest paths in~$G_{f_i}$ from~$s$ to nodes~$v \in V$ form a shortest path tree rooted at~$s$. Since the choice of the value~$z$ only influences the last augmentation of the algorithm, the algorithm performs the same augmentations when run for two different values $z_1 < z_2$ until the flow value~$|f_i|$ exceeds~$z_1$. We will exploit this observation in Lemma~\ref{lemma:cost function form}. 

Note that one could allow the cost function~$c$ to have negative values as well. As long as the network does not contain a cycle with negative total costs, the SSP algorithm is still applicable. However, as we cannot ensure this property if the edge costs are random variables, we made the assumption that all edge costs are non-negative.

\subsection{A Connection to the Integer Worst-case Bound}

We can concentrate on counting the number of augmenting steps of the SSP algorithm since each step can be implemented to run in time $O(m + n \log n)$ using Dijkstra's algorithm. Let us first consider the case that all edge costs are integers from $\SET{ 1, \ldots, C }$. In this case the length of any path in any residual network is bounded by~$nC$. We will see that the lengths of the augmenting paths are monotonically increasing. If there is no unique shortest path to augment flow along and ties are broken by choosing one with the fewest number of arcs, then the number of successive augmenting paths with the same length is bounded by $O(mn)$
(this follows from the analysis of the Edmonds-Karp algorithm for computing a maximum flow~\cite{CLRS}).
 Hence, the SSP algorithm terminates within $O(mn^2C)$ steps.

Now let us perturb the edge costs of such an integral instance independently by, for example, uniform additive noise from the interval $[-1,1]$.
This scenario is not covered by bounds for the integral case. Indeed, instances can be generated with positive probability for which the number of augmentation
steps is exponential in~$m$ and~$n$. Nevertheless, an immediate consequence of Theorem~\ref{maintheorem} is that, in expectation, the SSP
algorithm terminates within $O(mnC)$ steps on instances of this form.

\section{Terminology and Notation}

Consider the run of the SSP algorithm on the flow network~$G$. We denote the set $\SET{ f_0, f_1, \ldots }$ of all flows encountered by the SSP algorithm by~$\F_0(G)$. Furthermore, we set $\F(G) = \F_0(G) \setminus \SET{ f_0 }$.
(We omit the parameter $G$ if it is clear from the context.)

Let us remark that we have not specified in Algorithm~\ref{algorithm:SSP} which path is chosen if the shortest $s$-$t$ path is not unique. This is not important for our analysis because we will see in Section~\ref{sec:analysis} that this happens only with probability~$0$ in our probabilistic model. We can therefore assume $\F_0(G)$ to be well-defined.

By~$f_0$ and~$\fmax$, we denote the empty flow and the maximum flow, i.e., the flow that assigns~$0$ to all edges~$e$ and the flow of maximum value encountered by the SSP algorithm, respectively.

Let~$f_{i-1}$ and~$f_i$ be two consecutive flows encountered by the SSP algorithm and let~$P_i$ be the shortest path in the residual network~$G_{f_{i-1}}$, i.e., the SSP algorithm augments along~$P_i$ to increase flow~$f_{i-1}$ to obtain
flow~$f_i$. We call~$P_i$ the \emph{next path} of~$f_{i-1}$ and the \emph{previous path} of~$f_i$. To distinguish between the original network~$G$ and some residual network~$G_f$ in the remainder of this article, we refer to the edges in the residual network
as \emph{arcs}, whereas we refer to the edges in the original network
as \emph{edges}.

For a given arc~$e$ in a residual network~$G_f$, we denote by~$e_0$ the corresponding edge in the original network~$G$, i.e., $e_0 = e$ if $e \in E$ (i.e.~$e$ is a forward arc) and $e_0 = e^{-1}$ if $e \notin E$ (i.e.~$e$ is a backward arc). An arc~$e$ is called \emph{empty} (with respect to some residual network~$G_f$) if~$e$ belongs to~$G_f$, but $e^{-1}$ does not. Empty arcs~$e$ are either forward arcs that do not carry flow or backward arcs whose corresponding edge~$e_0$ carries as much flow as possible. We say that an arc \emph{becomes saturated} (during an augmentation) when it is contained in the current augmenting path, but it does not belong to the residual network that we obtain after this augmentation.

In the remainder, a \emph{path} is always a simple directed path.
Let~$P$ be a path, and let~$u$ and~$v$ be contained in~$P$ in this order. By $u \PATH[P] v$, we refer to the sub-path of~$P$ starting from node~$u$ going to node~$v$, by $\REVERSE{P}$ we refer to the path we obtain by reversing the direction of each edge of~$P$. We call any flow network~$G'$ a \emph{possible residual network} (of~$G$) if there is a flow~$f$ for~$G$ such that $G' = G_f$. Paths and cycles in possible residual networks are called \emph{possible paths} and \emph{possible cycles}, respectively. Let $\G = (V, E \cup E^{-1})$ for $E^{-1} = \SET{ e^{-1} \WHERE e \in E }$ denote the flow network that consists of all forward arcs and backward arcs. 

\section{Outline of Our Approach}

Our analysis of the SSP algorithm is based on the following idea: We identify a flow~$f_i \in \F_0$ with a real number by mapping~$f_i$ to the length~$\ell_i$ of the previous path~$P_i$ of~$f_i$. The flow~$f_0$ is identified with $\ell_0 = 0$. In this way,
we obtain a sequence $L = (\ell_0, \ell_1, \ldots)$ of real numbers. We show that this sequence is strictly monotonically increasing with probability~$1$. Since all costs are drawn from the interval $[0, 1]$, each element of~$L$ is from the interval $[0, n]$. To count the number of elements of~$L$, we partition the interval $[0, n]$ into small sub-intervals of length~$\e$ and sum up the number of elements of~$L$ in these intervals. By linearity of expectation, this approach carries over to the expected number of elements of~$L$. If~$\e$ is very small, then -- with sufficiently high probability -- each interval contains at most one element. 
If this is the case then it suffices to bound the probability that an element of~$L$ falls into some interval~$(d, d+\e]$ because
this probability equals the expected number of elements in~$(d, d+\e]$.

To do so, we assume for the moment that there is an integer~$i$ such that $\ell_i \in (d, d+\e]$.
By the previous assumption that for any interval of length~$\e$ there is at most
one path whose length is within this interval, we obtain that $\ell_{i-1} \leq
d$. We show that the augmenting path~$P_i$ uses an empty arc~$e$. Moreover, we
will see that we can reconstruct the flow~$f_{i-1}$ and the path~$P_i$ without knowing the costs of
edge~$e_0$ that corresponds to arc~$e$ in the original network.
This allows us to use the principle of deferred decisions: to bound the probability that~$\ell_i$ falls into the interval $(d, d+\e]$, we first reveal all costs~$c_{e'}$
with~$e'\neq e_0$. Then~$P_i$ is known and its length, which equals~$\ell_i$, can be expressed
as a linear function~$\kappa+c_{e_0}$ or~$\kappa-c_{e_0}$ for a known constant~$\kappa$. 
Consequently, the probability that~$\ell_i$ falls into the interval $(d, d+\e]$ is bounded by~$\e\phi$, as the probability density of $c_{e_0}$ is bounded by $\phi$.
Since the arc~$e$ is not always the same, we have to apply a union bound over all~$2m$ possible arcs. Summing up over all~$n/\e$ intervals the expected number of flows encountered by the SSP algorithm can be bounded by roughly $(n/\e) \cdot 2m \cdot \e\phi = 2mn\phi$.

There are some parallels to the analysis of the smoothed number of Pareto-optimal solutions in bicriteria linear optimization problems by Beier and V\"ocking~\cite{DBLP:journals/siamcomp/BeierV06}, although we have only one objective function.
In this context, we would call~$f_i$ the loser, $f_{i-1}$ the winner, and the difference~$\ell_i - d$ the loser gap. Beier and V\"ocking's analysis is also based on the observation
that the winner (which in their analysis is a Pareto-optimal solution and not a flow) can be reconstructed when all except for one random coefficients are revealed.
While this reconstruction is simple in the setting of bicriteria optimization problems, the reconstruction of the flow~$f_{i-1}$  in our setting is significantly more challenging and a main
difficulty in our analysis.

\section{Proof of the Upper Bound}
\label{sec:analysis}

Before we start with the analysis, note that due to our transformation of the general minimum-cost flow problem to a single-source-single-sink minimum-cost flow problem the cost perturbations only affect the original edges. The costs of the auxiliary edges are not perturbed but set to~$0$. Thus, we will slightly deviate from what we described in the outline by treating empty arcs corresponding to auxiliary edges separately.

The SSP algorithm is in general not completely specified, since at some point during the run of the algorithm there could exist multiple shortest $s$-$t$ paths in the residual network of the current flow.
The SSP algorithm then allows any of them to be chosen as the next augmenting path. Due to Lemma~\ref{lemma:different path lengths} and Property~\ref{property:different path lengths} we can assume that this is not the case in our setting and that the SSP algorithm is completely specified.

\begin{lemma}
\label{lemma:different path lengths}
For any real~$\e > 0$ the probability that there are two nodes~$u$ and~$v$ and two distinct possible $u$-$v$ paths whose lengths differ by at most~$\e$ is bounded from above by~$2n^{2n} \e \phi$.
\end{lemma}

\begin{proof}
Fix two nodes~$u$ and~$v$ and two distinct possible $u$-$v$ paths~$P_1$ and~$P_2$. Then there is an edge~$e$ such that one of the paths -- without loss of generality path~$P_1$ -- contains arc~$e$ or~$e^{-1}$, but the other one does not. If we fix all edge costs except the cost of edge~$e$, then the length of~$P_2$ is already determined whereas the length of~$P_1$ depends on the cost~$c_e$. Hence, $c_e$ must fall into a fixed interval of length~$2\e$ in order for the path lengths of~$P_1$ and~$P_2$ to differ by at most~$\e$. The probability for this is bounded by~$2\e \phi$ because~$c_e$ is chosen according to a density function that is bounded from above by~$\phi$. A union bound over all pairs~$(u, v)$ and all possible $u$-$v$ paths concludes the proof. 
\end{proof}

The proof also shows that we can assume that there is no $s$-$t$ path of length~$0$ and according to Lemma~\ref{lemma:different path lengths} we can assume that the following property holds since it holds with a probability of 1.

\begin{property}
\label{property:different path lengths}
For any nodes~$u$ and~$v$ the lengths of all possible $u$-$v$ paths are pairwise distinct.
\end{property}

\begin{lemma}
\label{lemma:distance monotonicity}
Let~$d_i(v)$ denote the distance from~$s$ to node~$v$ and $d'_i(v)$ denote the distance from node~$v$ to~$t$ in the residual network~$G_{f_i}$. Then
the sequences $d_0(v), d_1(v), d_2(v), \ldots$ and $d'_0(v), d'_1(v), d'_2(v), \ldots$ are monotonically increasing
for every~$v\in V$.
\end{lemma}

\begin{proof}
We only show the proof for the sequence $d_0(v), d_1(v), d_2(v), \ldots$. The proof for the sequence $d'_0(v), d'_1(v), d'_2(v), \ldots$ can be shown analogously.
Let $i \geq 0$ be an arbitrary integer. We show $d_i(v) \leq d_{i+1}(v)$ by induction on the depth of node~$v$ in the shortest path tree~$T_{i+1}$ of the residual network~$G_{f_{i+1}}$ rooted at~$s$. For the root~$s$, the claim holds since $d_i(s) = d_{i+1}(s) = 0$. Now assume that the claim holds for all nodes up to a certain depth~$k$, consider a node~$v$ with depth~$k+1$, and let~$u$ denote its parent. Consequently, $d_{i+1}(v) = d_{i+1}(u) + c_e$ for $e = (u, v)$. If arc~$e$ has been available in~$G_{f_i}$, then $d_i(v) \leq d_i(u) + c_e$. If not, then the SSP algorithm must have augmented along~$e^{-1}$ in step~$i+1$ to obtain flow~$f_{i+1}$ and, hence, $d_i(u) = d_i(v) + c_{e^{-1}} = d_i(v) - c_e$. In both cases the inequality $d_i(v) \leq d_i(u) + c_e$ holds. Applying the induction hypothesis for node~$u$, we obtain
$
  d_i(v)
  \leq d_i(u) + c_e
  \leq d_{i+1}(u) + c_e
  = d_{i+1}(v)
$.
\end{proof}

\begin{definition}
For a flow~$f_i \in \F_0$, we denote by $\spm[G](f_i)$ and $\spp[G](f_i)$ the length of the previous path~$P_i$ and the next path~$P_{i+1}$ of~$f_i$, respectively. By convention, we set $\spm[G](f_0) = 0$ and $\spp[G](\fmax) = \infty$. If the network~$G$ is clear from the context, then we simply write~$\spm(f_i)$ and~$\spp(f_i)$. By~$\COST$ we denote the cost function that maps reals~$x$ from the interval $\big[ 0, |\fmax| \big]$ to the cost of the cheapest flow~$f$ with value~$x$, i.e., $\COST(x) = \min \SET{ c(f) \WHERE |f| = x }$.
\end{definition}

The lengths~$\spm(f_i)$ correspond to the lengths~$\ell_i$ mentioned in the outline. The apparent notational overhead is necessary for formal correctness. In Lemma~\ref{lemma:cost function form}, we will reveal a connection between the values~$\spm(f_i)$ and the function~$\COST$. Based on this, we can focus on analyzing the function~$\COST$.

Lemma~\ref{lemma:distance monotonicity} implies in particular that the distance from the source~$s$
to the sink~$t$ is monotonically increasing, which yields the following corollary. 
\begin{corollary}
\label{corollary:path monotonicity}
Let~$f_i, f_j \in \F_0$ be two flows with $i < j$. Then $\spm(f_i) \leq \spm(f_j)$.
\end{corollary}


\begin{lemma}
\label{lemma:cost function form}
The function~$\COST$ is continuous, monotonically increasing, and piecewise linear, and the break points of the function are the values of the flows $f \in \F_0$ with $\spm(f) < \spp(f)$. For each flow $f \in \F_0$, the slopes of~$\COST$ to the left and to the right of~$|f|$ equal~$\spm(f)$ and~$\spp(f)$, respectively.
\end{lemma}

\begin{proof}
The proof follows from Theorem~\ref{thm:AllFlowsOpt} and the observation that the cost of the flow is linearly increasing when gradually increasing the flow along the shortest path in the residual network until at least one arc becomes saturated. The slope of the cost function is given by the length of that path. 
\end{proof}

\begin{example}
\label{example}
Consider the flow network depicted in Figure~\ref{fig:example network}. The cost~$c_e$ and the capacity~$u_e$ of an edge~$e$ are given by the notation~$c_e, u_e$. For each step of the SSP algorithm, Figure~\ref{exampletable} lists the relevant part of the augmenting path (excluding~$s$, $s'$, $t'$, and~$t$), its length, the amount of flow that is sent along that path, and the arcs that become saturated. As can be seen in the table, the values~$|f|$ of the encountered flows $f \in \F_0$ are $0$, $2$, $3$, $5$, $7$, $10$, and~$12$. These are the breakpoints of the cost function~$\COST$, and the lengths of the augmenting paths equal the slopes of~$\COST$ (see Figure~\ref{fig:example cost function}).

\begin{fig}
  \begin{minipage}{0.54\textwidth}
  	\GFX[width=\textwidth]{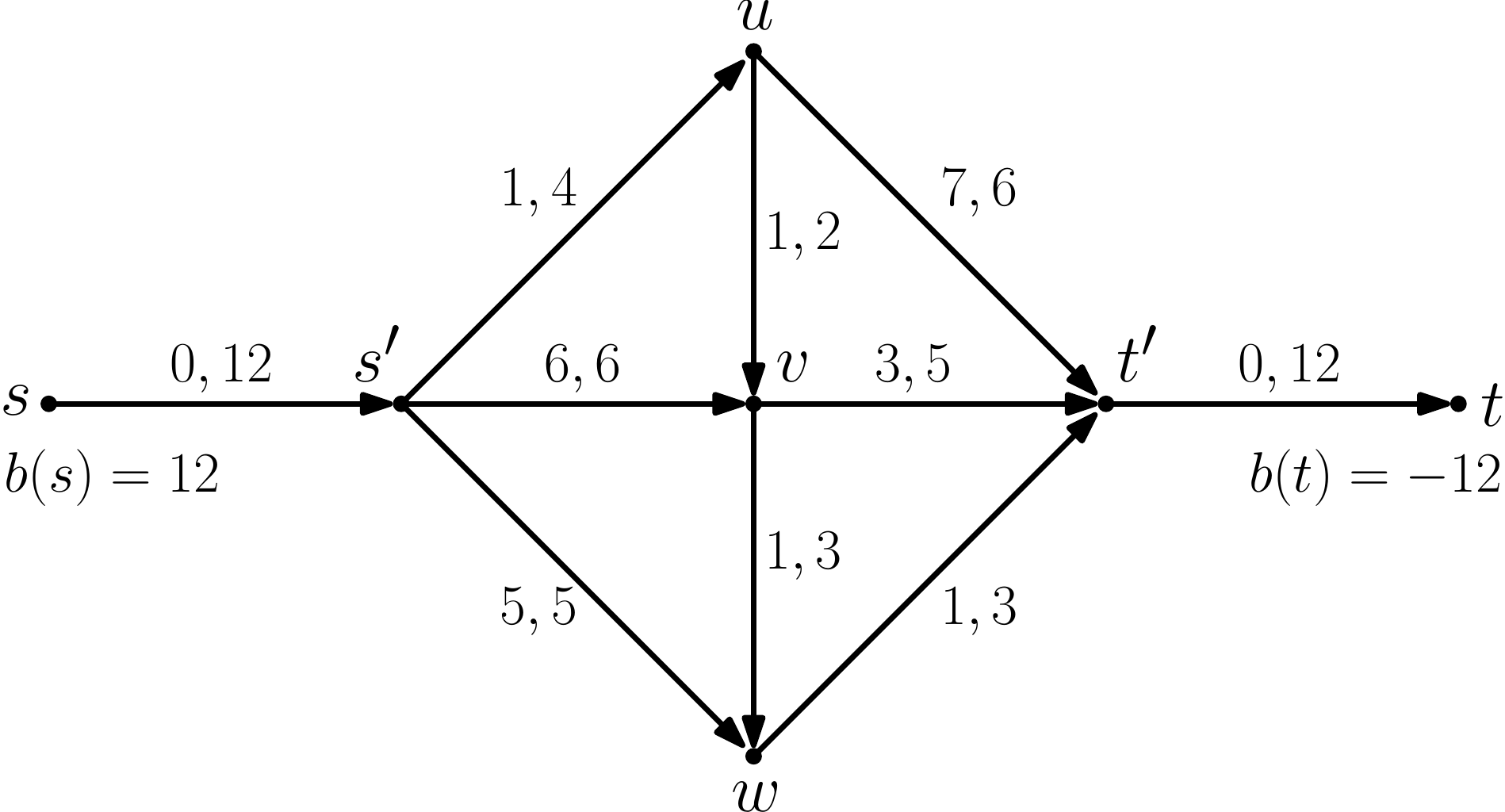}
  	\caption{Minimum-cost flow network with master source~$s$ and master sink~$t$.}
  	\label{fig:example network}
  \end{minipage}
  \qquad
  \begin{minipage}{0.4\textwidth}
	\GFX[width=\textwidth]{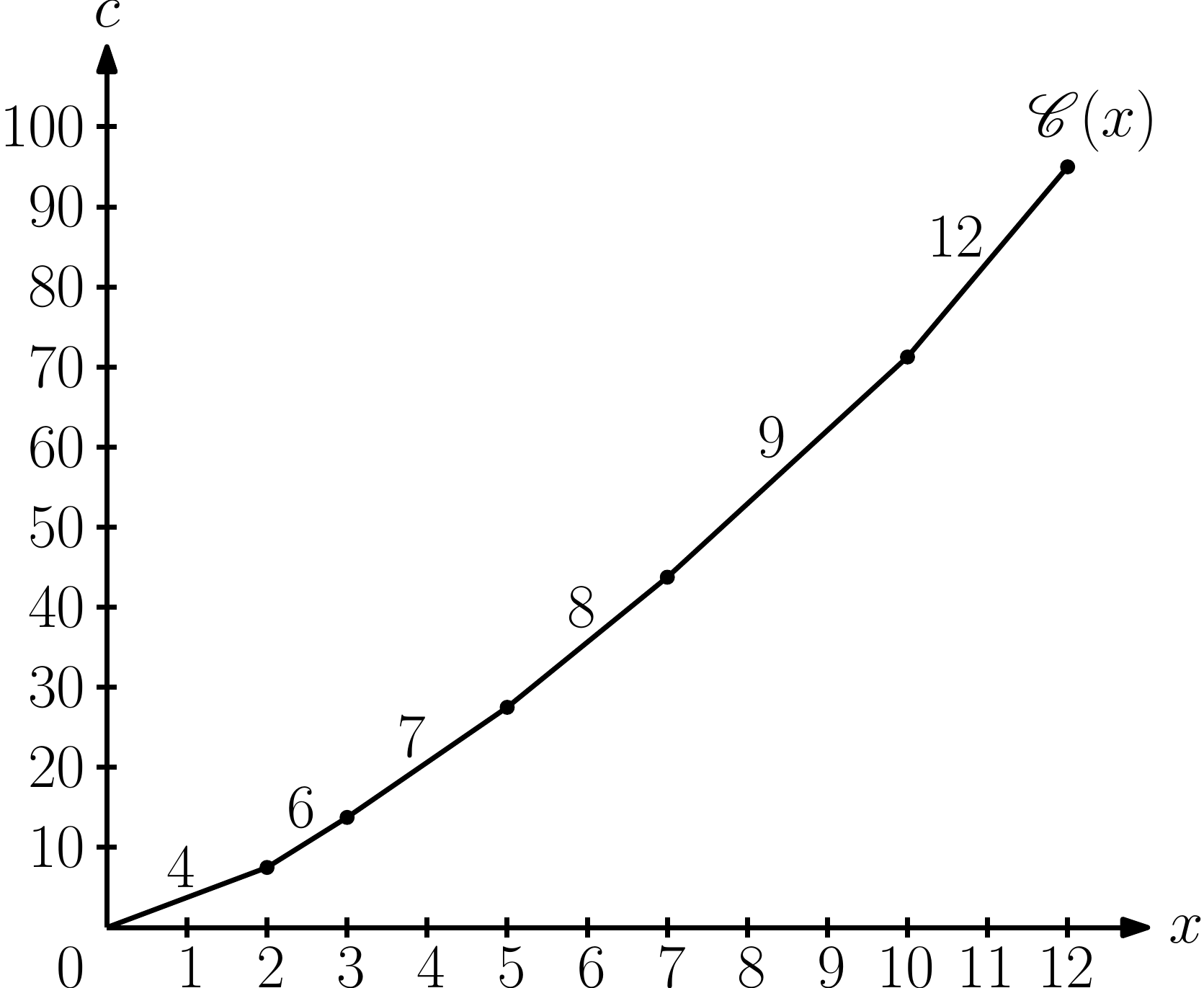}
  	\caption{Cost function~$\COST$.}
  	\label{fig:example cost function}
  \end{minipage}\\\vspace{0.5cm}
  \begin{minipage}{\textwidth}
\centering\newcommand{\CENTER}{@{\hspace{0.5ex}}c@{\hspace{0.5ex}}}
  \begin{tabular}{|@{\hspace{0.5ex}}l@{\hspace{0.5ex}}|\CENTER|\CENTER|\CENTER|\CENTER|\CENTER|\CENTER|}
    \hline
    \textbf{step}           & $1$       & $2$       & $3$      & $4$       & $5$       & $6$      \cr
    \hline
    \textbf{path}           & $u, v, w$ & $w$       & $w, v$   & $u$       & $v$       & $v, u$   \cr
    \hline
    \textbf{path length}    & $4$       & $6$       & $7$      & $8$       & $9$       & $12$     \cr
    \hline
    \textbf{amount of flow} & $2$       & $1$       & $2$      & $2$       & $3$       & $2$      \cr
    \hline
    \textbf{saturated arcs} & $(u, v)$  & $(w, t')$ & $(w, v)$ & $(s', u)$ & $(v, t')$ & $(v, u)$ \cr
    \hline
  \end{tabular}
\caption{The augmenting paths for Example~\ref{example}.}
\label{exampletable}    
  \end{minipage}
\end{fig}

\end{example}

With the following definition, we lay the foundation for distinguishing between original edges with perturbed costs and auxiliary edges whose costs are set to~$0$.

\begin{definition}
Let $f \in \F_0$ be an arbitrary flow. An empty arc~$e$ in the residual network~$G_f$ that does not correspond to an auxiliary edge is called a \emph{good arc}. We call~$f$ a \emph{good flow} if $f \neq f_0$ and if the previous path of~$f$ contains a good arc in the previous residual network. Otherwise, $f$ is called a \emph{bad flow}.
\end{definition}

Before we can derive a property of good arcs that are contained in the previous path of good flows, we need to show that for each flow value the minimum-cost flow is unique with probability~$1$.

\begin{lemma}
\label{lemma:non zero cycle lengths}
For any real~$\e > 0$ the probability that there exists a possible cycle whose costs lie in~$[0,\e]$ is bounded from above by~$2n^{2n} \e \phi$.
\end{lemma}

\begin{proof}
Assume that there exists a cycle~$K$ whose costs lie in~$[0,\e]$. Then~$K$ contains two nodes~$u$ and~$v$ and consists of a $u$-$v$ path~$P_1$ and a $v$-$u$ path~$P_2$.
Then $P_1$ and $\stackrel{\leftarrow}{P_2}$ are two distinct $u$-$v$ paths. Since~$K$ has costs in~$[0,\e]$, the costs of~$P_1$ and $\stackrel{\leftarrow}{P_2}$ differ by at most~$\e$.
Now Lemma~\ref{lemma:different path lengths} concludes the proof.
\end{proof}

According to Lemma~\ref{lemma:non zero cycle lengths} we can assume that the following property holds since it holds with a probability of 1.

\begin{property}
\label{property:non zero cycle lengths}
There exists no possible cycle with costs~$0$.
\end{property}

With Property~\ref{property:non zero cycle lengths} we can show that the minimum-cost flow is unique for each value.

\begin{lemma}
\label{lemma:unique minimum cost flow}
For each value $B \in \RR_{\geq 0}$ there either exists no flow~$f$ with $|f| = B$ or there exists a unique minimum-cost flow~$f$ with $|f| = B$.
\end{lemma}

\begin{proof}
Assume that there exists a value $B \in \RR_{\geq 0}$ and two distinct minimum-cost flows $f$ and $f'$ with $|f| = |f'| = B$. Let $E_{\Delta} := \{e \in E \mid f_e \neq f'_e\}$ be the set of edges on which $f$ and $f'$ differ. We show in the following that the set $E_{\Delta}$ contains at least one undirected cycle~$K$.
Since $f$ and $f'$ are distinct flows, the set $E_{\Delta}$ cannot be empty. For $v \in V$, let us denote by $f_-(v) = \sum_{e=(u,v) \in E} f_e$ the flow entering $v$ and by $f_+(v) = \sum_{e=(v,w) \in E} f_e$ the flow going out of~$v$ ($f'_-(v)$ and $f'_+(v)$ are defined analogously). Flow conservation and $|f| = |f'|$ imply $f_-(v) - f'_-(v) = f_+(v) - f'_+(v)$ for all $v \in V$. 
Now let us assume $E_{\Delta}$ does not contain an undirected cycle. In this case there must exist a vertex~$v \in V$ with exactly one incident edge in~$E_{\Delta}$. We will show that this cannot happen.

Assume $f_-(v) - f'_-(v) \neq 0$ for some $v \in V$. Then the flows $f$ and $f'$ differ on at least one edge~$e = (u,v) \in E$. Since this case implies $f_+(v) - f'_+(v) \neq 0$, they also differ on at least one edge~$e' = (v,w) \in E$ and both these edges belong to~$E_{\Delta}$. It remains to consider nodes~$v \in V$ with $f_-(v) - f'_-(v) = f_+(v) - f'_+(v) = 0$ and at least one incident edge in~$E_{\Delta}$. For such a node~$v$ there exists an edge $e = (u,v) \in E$ (or $e = (v,w) \in E$) with $f_e \neq f'_e$. It follows $\sum_{e'=(u',v) \in E, e' \neq e} f_{e'} - f'_{e'} \neq 0$ (or $\sum_{e'=(v,w') \in E, e' \neq e} f_{e'} - f'_{e'} \neq 0$) which implies that there exists another edge $e' = (u',v) \neq e$ (or $e = (v,w') \neq e$) with $f_{e'} \neq f'_{e'}$.

For the flow~$f'' = \frac{1}{2} f + \frac{1}{2} f'$, which has the same costs as~$f$ and~$f'$ and is hence a minimum-cost flow with $|f''|=B$ as well, we have $f''(e) \in (0,u_e)$ for all $e \in E_{\Delta}$. The flow~$f''$ can therefore be augmented in both directions along $K$. Due to Property~\ref{property:non zero cycle lengths}, augmenting $f''$ in one of the two directions along~$K$ will result in a better flow. This is a contradiction.
\end{proof}

Now we derive a property of good arcs that are contained in the previous path of good flows.
This property allows us to bound the probability that one of the lengths $\spm(f_i)$ falls into a given interval of length~$\e$.

\begin{lemma}
\label{lemma:cost monotonicity}
Let~$f \in \F_0$ be a predecessor of a good flow for which $\spm[G](f) < \spp[G](f)$ holds 
Additionally, let~$e$ be a good arc in the next path of~$f$, and let~$e_0$ be the edge in~$G$ that corresponds to~$e$. Now change the cost of~$e_0$ to $c'_{e_0} = 1$ ($c'_{e_0} = 0$) if $e_0 = e$ ($e_0 = e^{-1}$), i.e., when $e$ is a forward (backward) arc. In any case, the cost of arc~$e$ increases. We denote the resulting flow network by~$G'$. Then $f \in \F_0(G')$. Moreover, the inequalities
$
  \spm[G'](f) \leq \spm[G](f) 
	< \spp[G](f) \leq \spp[G'](f)
$
hold.
\end{lemma}

\begin{proof}
Let~$\COST$ and~$\COST'$ be the cost functions of the original network~$G$ and the modified network~$G'$, respectively.
Both functions are of the form described in Lemma~\ref{lemma:cost function form}. In particular, they are continuous and the breakpoints correspond to the values of the flows $\tilde{f} \in \F_0(G)$ and $\hat{f} \in \F_0(G')$ with $\spm[G](\tilde{f}) < \spp[G](\tilde{f})$ and $\spm[G'](\hat{f}) < \spp[G'](\hat{f})$, respectively.

We start with analyzing the case~$e_0 = e$. In this case, we set $\COST'' = \COST'$ and observe that increasing the cost of edge~$e_0$ to~$1$ cannot decrease the cost of any flow in $G$. Hence, $\COST'' \geq \COST$. Since flow~$f$ does not use arc~$e$, its costs remain unchanged, i.e., $\COST''(|f|) = \COST(|f|)$.

If~$e_0 = e^{-1}$, then we set $\COST'' = \COST' + \Delta_{e_0}$ for $\Delta_{e_0} = u_{e_0} \cdot c_{e_0}$. This function is also piecewise linear and has the same breakpoints and slopes as~$\COST'$. Since the flow on edge~$e_0$ cannot exceed the capacity~$u_{e_0}$ of edge~$e_0$ and since the cost on that edge has been reduced by~$c_{e_0}$ in~$G'$, the cost of each flow is reduced by at most~$\Delta_{e_0}$ in~$G'$. Furthermore, this gain is only achieved for flows that entirely use edge~$e_0$ like~$f$ does. Hence, $\COST'' \geq \COST$ and $\COST''(|f|) = \COST(|f|)$.

\begin{GFXFIG}[width=0.5\textwidth]{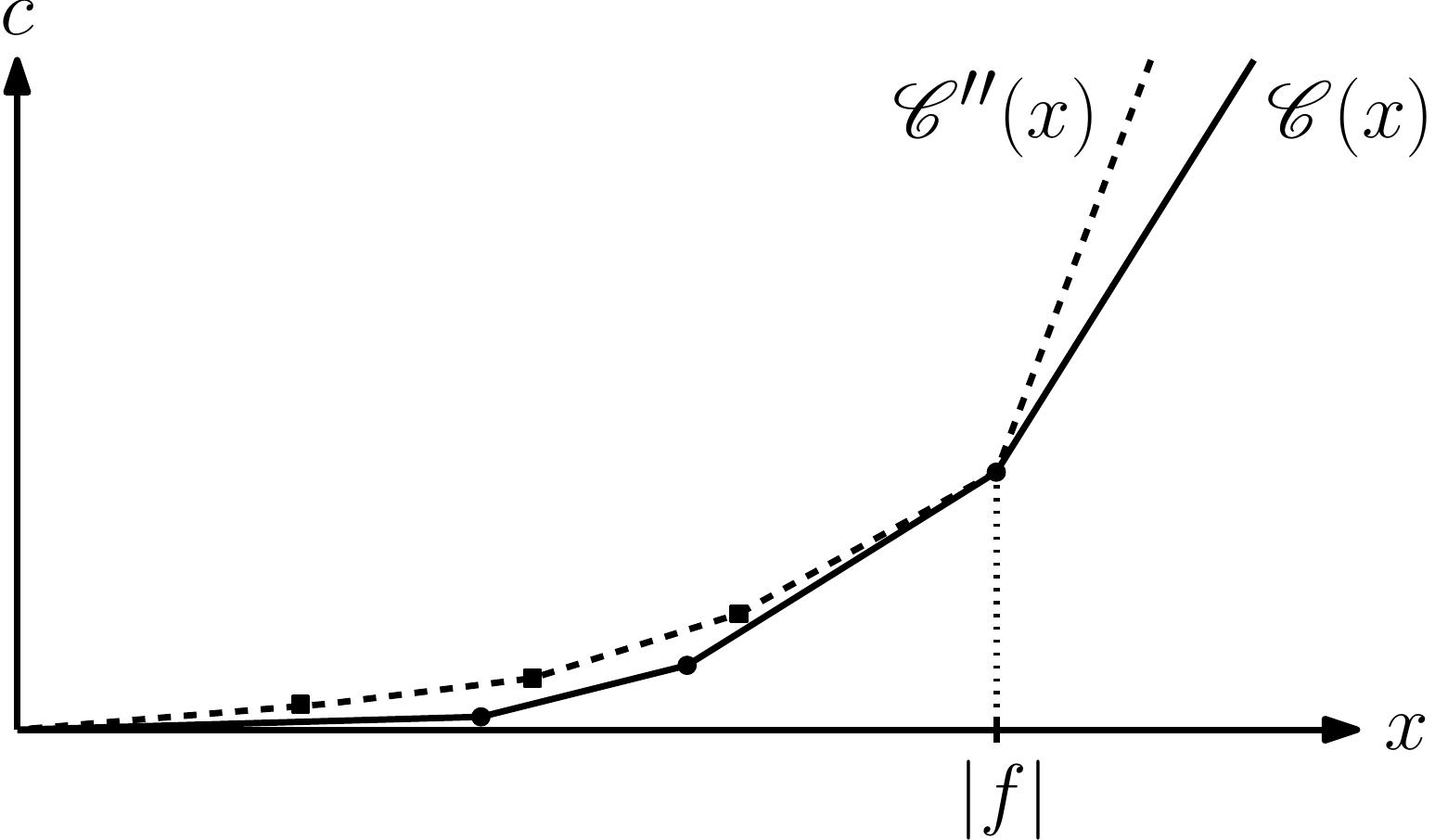}
\caption{Cost function~$\COST$ and function~$\COST''$.}
\label{fig:cost monotonicity}
\end{GFXFIG}

Due to $\COST'' \geq \COST$, $\COST''(|f|) = \COST(|f|)$, and the form of both
functions, the left-hand derivative of~$\COST''$ at~$|f|$ is at most the
left-hand derivative of~$\COST$ at~$|f|$ (see Figure~\ref{fig:cost
monotonicity}). Since~$|f|$ is a breakpoint of~$\COST$, this implies that~$|f|$
is also a breakpoint of~$\COST''$ and that the slope of~$\COST''$ to the left 
of~$|f|$ is at most the slope of~$\COST$ to the left of~$|f|$.
For the same reasons, the right-hand derivative of~$\COST''$ at~$|f|$ is at
least the right-hand derivative of~$\COST$ at~$|f|$ and the slope 
of~$\COST''$ to the right of~$|f|$ is at least the slope of~$\COST$ to the right of~$|f|$.
These properties carry over to~$\COST'$. Hence, $\F_0(G')$ contains a flow~$f'$ with $|f'| = |f|$. Since $f$ is a minimum-cost flow with respect to $c$, $f'$ is a minimum-cost flow with respect to $c'$, we have $c'(f) = c(f)$ and $c'(f^*) \geq c(f^*)$ for all possible flows~$f^*$, Lemma~\ref{lemma:unique minimum cost flow} yields $f = f'$ and therefore $f \in \F_0(G')$. Recalling the fact that the slopes correspond
to shortest $s$-$t$ path lengths, the stated chain of inequalities follows.
\end{proof}

Lemma~\ref{lemma:cost monotonicity} suggests Algorithm~\ref{reconstruct} ($\Reconstruct$) for reconstructing a flow~$f$ based on a good arc~$e$ that belongs to the shortest path in the residual network~$G_f$ and on a threshold $d \in \big[ \spm(f), \spp(f) \big)$. The crucial fact that we will later exploit is that for this reconstruction the cost~$c_{e_0}$ of edge~$e_0$ does not have to be known.
(Note that we only need $\Reconstruct$ for the analysis in order to show that the flow~$f$ can be reconstructed.)

\begin{algorithm}[t]
  \caption{$\Reconstruct(e, d)$.}
  \begin{algorithmic}[1]
    \STATE let~$e_0$ be the edge that corresponds to arc~$e$ in the original network~$G$

    \STATE change the cost of edge~$e_0$ to $c'_{e_0} = 1$ if~$e$ is a forward arc or to $c'_{e_0} = 0$ if~$e$ is a backward arc

    \STATE start running the SSP algorithm on the modified network~$G'$

    \STATE stop when the length of the shortest $s$-$t$ path in the residual network of the current flow~$f'$ exceeds~$d$

    \STATE output~$f'$
  \end{algorithmic}
  \label{reconstruct}
\end{algorithm}

\begin{corollary}
\label{corollary:flow reconstruction}
Let~$f \in \F_0$ be a predecessor of a good flow, let~$e$ be a good arc in the next path of~$f$, and let $d \in \big[ \spm(f), \spp(f) \big)$ be a real number. Then $\Reconstruct(e, d)$ outputs flow~$f$.
\end{corollary}

\begin{proof}
By applying Lemma~\ref{lemma:cost monotonicity}, we obtain $f \in \F_0(G')$ and $\spm[G'](f) \leq d < \spp[G'](f)$. Together with Corollary~\ref{corollary:path monotonicity}, this implies that $\Reconstruct(e, d)$ does not stop before encountering flow~$f$ and stops once it encounters~$f$. Hence, $\Reconstruct(e, d)$ outputs flow~$f$. 
\end{proof}

Corollary~\ref{corollary:flow reconstruction} is an essential component of the proof of Theorem~\ref{maintheorem} but it only describes how to reconstruct predecessor flows~$f$ of good flows with $\spm(f) < \spp(f)$. In the next part of this section we show that
most of the flows are
good flows and that, with a probability of~$1$, the inequality $\spm(f) < \spp(f)$ holds for any flow~$f \in \F_0$.

\begin{lemma}
\label{lemma:one empty arc}
In any step of the SSP algorithm, any $s$-$t$ path in the residual network contains at least one empty arc.
\end{lemma}

\begin{proof}\renewcommand{\P}{P'}
The claim is true for the empty flow~$f_0$. Now consider a flow~$f_i \in \F$, its predecessor flow~$f_{i-1}$, the path~$P_i$, which is a shortest path in the residual network~$G_{f_{i-1}}$, and an arbitrary $s$-$t$ path~$P$ in the current residual network~$G_{f_i}$. We show that at least one arc in~$P$ is empty.

For this, fix one arc~$e = (x, y)$ from~$P_i$ that is not contained in the current residual network~$G_{f_i}$ since it became saturated by the augmentation along~$P_i$. Let~$v$ be the first node of~$P$ that occurs in the sub-path~$y \PATH[P_i] t$ of~$P_i$, and let~$u$ be the last node in the sub-path~$s \PATH[P] v$ of~$P$ that belongs to the sub-path~$s \PATH[P_i] x$ of~$P_i$ (see Figure~\ref{fig:one empty arc}). By the choice of~$u$ and~$v$, all nodes on the sub-path~$\P = u \PATH[P] v$ of~$P$ except~$u$ and~$v$ do not belong to~$P_i$. Hence, the arcs of~$\P$ are also available in the residual network~$G_{f_{i-1}}$ and have
the same capacity in both residual networks~$G_{f_{i-1}}$ and~$G_{f_i}$.

\begin{GFXFIG}[width=0.6\textwidth]{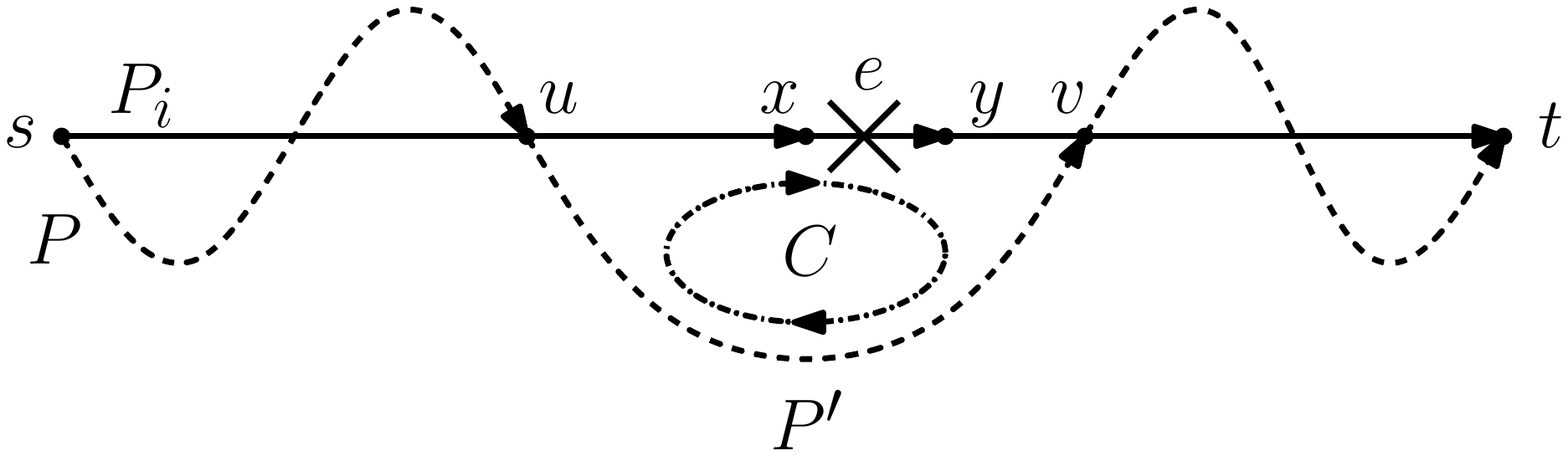}
\caption{Paths~$P$ and~$P_i$ in the residual network~$G_{f_i}$.}
\label{fig:one empty arc}
\end{GFXFIG}

In the remainder of this proof,
we show that at least one arc of~$\P$ is empty. Assume to the contrary
that none of the arcs is empty in~$G_{f_i}$ and, hence, in~$G_{f_{i-1}}$. This implies that, for each arc~$e \in \P$, the residual network~$G_{f_{i-1}}$ also contains the arc~$e^{-1}$. Since~$P_i$ is the shortest $s$-$t$ path in~$G_{f_{i-1}}$ and since the lengths of all possible $s$-$t$ paths are pairwise distinct, the path~$s \PATH[P_i] u \PATH[P] v \PATH[P_i] t$ is longer than~$P_i$. Consequently, the path~$\P = u \PATH[P] v$ is longer than the path~$u \PATH[P_i] v$. This contradicts the fact that flow~$f_{i-1}$ is optimal since the arcs of path~$u \PATH[P_i] v$ combined with the reverse arcs~$e^{-1}$ of all the arcs~$e$ of path~$\P$ form a directed cycle~$C$ in~$G_{f_{i-1}}$ of negative costs. 
\end{proof}

We want to partition the interval $[0, n]$ into small sub-intervals of length~$\e$ and treat the number of lengths $\spm(f_i)$ that fall into a given sub-interval as a binary random variable. This may be wrong if there are two possible $s$-$t$ paths whose lengths differ by at most~$\e$. In this case whose probability tends to $0$ (see Lemma~\ref{lemma:different path lengths}) we will simply bound the number of augmentation steps of the SSP algorithm by a worst-case bound according to the following lemma.

\begin{lemma}
\label{lemma:worst case number}
The number~$|\F_0|$ of flows encountered by the SSP algorithm is bounded by~$3^{m+n}$. 
\end{lemma}

\begin{proof}
We call two possible residual networks equivalent if they contain the same arcs. Equivalent possible residual networks have the same shortest $s$-$t$ path in common. The length of this path is also the same. Assume that for two distinct flows $f_i, f_j \in \F_0$ with $i < j$, the residual networks~$G_{f_i}$ and~$G_{f_j}$ are equivalent. We then have $\spm(f_{i+1}) = \spp(f_i)$ = $\spp(f_j) =\spm(f_{j+1})$ and due to Corollary~\ref{corollary:path monotonicity}, $\spm(f_{i+1}) = \spm(f_k) =\spm(f_{j+1})$ for all $i < k \leq j+1$. Property~\ref{property:different path lengths} then implies $P_{i+1} = P_k$ for all $i < k \leq j+1$ and especially $P_{i+1} = P_{i+2}$, which is a contradiction. Therefore the number of equivalence classes is bounded by~$3^{m+n}$ since there are~$m$ original edges and at most~$n$ auxiliary edges. This completes the proof. 
\end{proof}

\begin{lemma}
\label{lemma:upper bound bad flows}
There are at most~$n$ bad flows~$f \in \F$.
\end{lemma}

\begin{proof}
According to Lemma~\ref{lemma:one empty arc}, the augmenting path contains an empty arc~$e$
in each step. If~$e$ is an arc that corresponds to an auxiliary edge (this is the only case when~$e$ is not a good arc), then~$e$ is not empty after the augmentation. Since the SSP algorithm does not augment along arcs~$e^{-1}$ if~$e$ is an arc that corresponds to an auxiliary edge, non-empty arcs that correspond to auxiliary edges cannot be empty a second time. Thus, there can be at most~$n$ steps where the augmenting path does not contain a good arc. This implies that there are at most~$n$ bad flows~$f \in \F$. 
\end{proof}

We can now bound the probability that there is a flow~$f_i \in \F$ whose previous path's length $\spm(f_i)$ falls into a given sub-interval of length~$\e$. Though we count bad flows separately, they also play a role in bounding the probability that there is a \emph{good} flow~$f_i \in \F$ such that $\spm(f_i)$ falls into a given sub-interval of length~$\e$.

\begin{lemma}
\label{lemma:loser gap}
For a fixed real $d \geq 0$, let~$\event_{d, \e}$ be the event that there is a flow~$f \in \F$ for which $\spm(f) \in (d, d+\e]$, and let~$B_{d, \e}$ be the event that there is a bad flow~$f' \in \F$ for which $\spm(f') \in (d, d+\e]$. Then the probability of~$\event_{d, \e}$ can be bounded by
$
  \Pr{\event_{d, \e}} \leq 2m \e \phi + 2 \cdot \Pr{B_{d, \e}}
$.
\end{lemma}

\begin{proof}
Let $A_{d, \e}$ be the event that there is a good flow~$f \in \F$ for which $\spm(f) \in (d, d+\e]$. Since $\event_{d, \e} = A_{d, \e} \cup B_{d, \e}$, it suffices to show that $\Pr{A_{d, \e}} \leq 2m \e \phi + \Pr{B_{d, \e}}$. Consider the event that there is a good flow whose previous path's length lies in the interval $(d, d+\e]$. Among all these good flows, let~$\hat{f}$ be the one with the smallest value $\spm(\hat{f})$, i.e., $\hat{f}$ is the first good flow~$f$ encountered by the SSP algorithm for which $\spm(f) \in (d, d+\e]$, and let~$f^*$ be its previous flow. Flow~$f^*$ always exists since~$\hat{f}$ cannot be the empty flow~$f_0$. Corollary~\ref{corollary:path monotonicity} and Property~\ref{property:different path lengths} yield $\spm(f^*) < \spm(\hat{f})$. Thus, there can only be two cases: If $\spm(f^*) \in (d, d+\e]$, then~$f^*$ is a bad flow by the choice of~$\hat{f}$ and, hence, event~$B_{d, \e}$ occurs. The interesting case, which we consider now, is when $\spm(f^*) \leq d$ holds. If this is true, then $d \in [\spm(f^*), \spp(f^*))$ due to $\spp(f^*) = \spm(\hat{f})$.

As~$\hat{f}$ is a good flow, the shortest path in the residual network~$G_{f^*}$ contains a good arc $e = (u, v)$. Applying Corollary~\ref{corollary:flow reconstruction} we obtain that we can reconstruct flow~$f^*$ by calling $\Reconstruct(e, d)$. The shortest $s$-$t$ path~$P$ in the residual network~$G_{f^*}$ is the previous path of~$\hat{f}$ and its length equals~$\spm(\hat{f})$. Furthermore, $P$ is of the form $s \PATH[P] u \to v \PATH[P] t$, where $s \PATH[P] u$ and $v \PATH[P] t$ are shortest paths in~$G_{f^*}$ from~$s$ to~$u$ and from~$v$ to~$t$, respectively. These observations yield
\[
  A_{d, \e} \subseteq \bigcup_{e \in E} R_{e, d, \e} \cup \bigcup_{e \in E} R_{e^{-1}, d, \e} \cup B_{d, \e} \COMMA
\]
where~$R_{e, d, \e}$ for some arc~$e = (u, v)$ denotes the following event: The event~$R_{e, d, \e}$ occurs if $\ell \in (d, d+\e]$, where~$\ell$ is the length of the shortest $s$-$t$ path that uses arc~$e$ in~$G_f$, the residual network of the flow~$f$ obtained by calling the procedure $\Reconstruct(e, d)$. Therefore, the probability of event~$A_{d, \e}$ is bounded by
\[
  \sum_{e \in E} \Pr{R_{e, d, \e}} + \sum_{e \in E} \Pr{R_{e^{-1}, d, \e}} + \Pr{B_{d, \e}} \DOT
\]
We conclude the proof by showing $\Pr{R_{e, d, \e}} \leq \e \phi$. For this, let~$e_0$ be the edge corresponding to arc~$e = (u, v)$ in the original network. If we fix all edge costs except cost~$c_{e_0}$ of edge~$e_0$, then the output~$f$ of $\Reconstruct(e, d)$ is already determined. The same holds for the shortest $s$-$t$ path in~$G_f$ that uses arc~$e$ since it is of the form $s \PATH u \to v \PATH t$ where $P_1 = s \PATH u$ is a shortest $s$-$u$ path in~$G_f$ that does not use~$v$ and where $P_2 = v \PATH t$ is a shortest $v$-$t$ path in~$G_f$ that does not use~$u$. The length~$\ell$ of this path, however, depends linearly on the cost~$c_{e_0}$. To be more precise, $\ell = \ell' + c_e = \ell' + \sgn(e) \cdot c_{e_0}$, where $\ell'$ is the length of~$P_1$ plus the length of~$P_2$ and where
\[ 
\sgn(e) = \begin{cases}
    +1 & \tIF\ e_0 = e \COMMA \cr
    -1 & \tIF\ e_0 = e^{-1} \DOT
  \end{cases}
\]
Hence, $\ell$ falls into the interval $(d, d+\e]$ if and only if~$c_{e_0}$ falls into some fixed interval of length~$\e$. The probability for this is bounded by $\e \phi$ as~$c_{e_0}$ is drawn according to a distribution whose density is bounded by~$\phi$. 
\end{proof}

\begin{corollary}
\label{corollary:expected number of SSP steps}
The expected number of augmentation steps the SSP algorithm performs is bounded by $2mn\phi + 2n$.
\end{corollary}

\begin{proof}
Let~$X = |\F|$ be the number of augmentation steps of the SSP algorithm. For reals $d, \e > 0$, let~$\event_{d, \e}$ and $B_{d, \e}$ be the events defined in Lemma~\ref{lemma:loser gap}, let~$X_{d,\e}$ be the number of flows $f \in \F$ for which $\spm(f) \in (d, d+\e]$, and let $Z_{d, \e} = \min \SET{ X_{d, \e}, 1 }$ be the indicator variable of event~$\event_{d, \e}$.

Since all costs are drawn from the interval $[0, 1]$, the length of any possible $s$-$t$ path is bounded by~$n$. Furthermore, according to Corollary~\ref{corollary:path monotonicity}, all lengths are non-negative (and positive with a probability of 1). Let~$F_\e$ denote the event that there are two possible $s$-$t$ paths whose lengths differ by at most~$\e$. Then, for any positive integer~$k$, we obtain
\[
  X
  = \sum_{i=0}^{k-1} X_{i \cdot \frac{n}{k}, \frac{n}{k}}
  \begin{cases}
    = \sum \limits_{i=0}^{k-1} Z_{i \cdot \frac{n}{k}, \frac{n}{k}} & \text{if $F_{\frac{n}{k}}$ does not occur} \COMMA \cr
    \leq 3^{m+n} & \text{if $F_{\frac{n}{k}}$ occurs} \DOT
  \end{cases}
\]
Consequently,
\begin{align*}
  \Ex{X}
  &\leq \sum_{i=0}^{k-1} \Ex{Z_{i \cdot \frac{n}{k}, \frac{n}{k}}} + 3^{m+n} \cdot \Pr{F_{\frac{n}{k}}} \cr
  &= \sum_{i=0}^{k-1} \Pr{\event_{i \cdot \frac{n}{k}, \frac{n}{k}}} + 3^{m+n} \cdot \Pr{F_{\frac{n}{k}}} \cr
  &\leq 2mn\phi + 2 \cdot \sum_{i=0}^{k-1} \Pr{B_{i \cdot \frac{n}{k}, \frac{n}{k}}} + 3^{m+n} \cdot \Pr{F_{\frac{n}{k}}} \cr
  &\leq 2mn\phi + 2n + 3^{m+n} \cdot \Pr{F_{\frac{n}{k}}} \DOT
\end{align*}
The second inequality is due to Lemma~\ref{lemma:loser gap} whereas the third inequality stems from Lemma~\ref{lemma:upper bound bad flows}. The claim follows since $\Pr{F_{\frac{n}{k}}} \to 0$ for $k \to \infty$ in accordance with Lemma~\ref{lemma:different path lengths}.
\end{proof}

Now we are almost done with the proof of our main theorem.

\begin{proof}
Since each step of the SSP algorithm runs in time $O(m + n \log n)$ using Dijkstra's algorithm (see, e.g.,
Korte~\cite{Korte:2007:COT:1564997} for details), applying Corollary~\ref{corollary:expected number of SSP steps} yields the desired result. 
\end{proof}

\section{Proof of the Lower Bound}
\label{sec:lower bound}



This section is devoted to the proof of Theorem~\ref{theorem:lower bound}. For
given positive integers~$n$, $m \in \sSET{ n, \ldots, n^2 }$, and $\phi \le 2^n$
let $k = \floor{\log_2 \phi} - 5 = O(n)$ and $M = \min \SET{ n,
2^{\floor{\log_2 \phi}}/4-2 } = \Theta(\min \sSET{ n, \phi })$. In the following we assume that $\phi \geq 64$, such that we have $k,M \geq 1$. If $\phi<64$, the lower bound on the number of augmentation steps from Theorem~\ref{theorem:lower bound} reduces to $\Omega(m)$ and a simple flow network like the network~$G_1$, as explained below, which we will use as initial network in case $\phi \geq 64$, with $O(n)$ nodes, $O(m)$ edges, and uniform edge costs proves the lower bound.  

We construct a flow
network with $2n+2k+2+4M = O(n)$ nodes, $m+2n+4k-4+8M = O(m)$ edges, 
and smoothing parameter~$\phi$ on
which the SSP algorithm requires $m \cdot 2^{k-1} \cdot 2M = \Theta(m \cdot
\phi \cdot \min \sSET{ n, \phi })$ augmentation steps in expectation. To be
exact, we show that for any realization of the edge costs 
for which there do not exist multiple paths with exactly the same costs (Property~\ref{property:different path lengths}) 
the SSP algorithm requires that many iterations. Since this happens with probability~1,
we will assume in the following that Property~\ref{property:different path lengths} 
holds without further mention. 

For the sake of simplicity we consider edge cost densities $f_e \colon [0, \phi]
\to [0, 1]$ instead of $f_e \colon [0, 1] \to [0, \phi]$. This is an equivalent
smoothed input model because both types of densities can be transformed into
each other by scaling by a factor of~$\phi$ and because the behavior of the SSP
algorithm is invariant under scaling of the edge costs. Furthermore, our
densities~$f_e$ will be uniform distributions on intervals~$I_e$ with lengths of
at least~$1$. In the remainder of this section we only construct these
intervals~$I_e$. Also, all minimum-cost flow networks constructed in this
section have a unique source node~$s$ and a unique sink node~$t$, which is
always clear from the context. The balance values of the nodes are defined as
$b(v) = 0$ for all nodes~$v \notin \SET{ s, t }$ and $-b(t) = b(s) = \sum_{e =
(s, v)} u_e = \sum_{e = (w, t)} u_e$, that is, each $b$-flow equals a maximum
$s$-$t$-flow.

The construction of the desired minimum-cost flow network~$G$ consists of three steps,
which we sketch below and describe in more detail thereafter. Given Property~\ref{property:different path lengths}, our choice of distributions for the edge costs ensures that the behavior of the SSP algorithm is the same for every realization of the edge costs.
\begin{enumerate}
\item In the first step we define a simple flow network~$G_1$ with a source~$s_1$ and a sink~$t_1$ on which the SSP algorithm requires~$m$ augmentation steps. 

\item In the second step we take a flow network~$G_i$, starting with $i = 1$, as the basis for constructing a larger flow network~$G_{i+1}$. We obtain the new flow network by adding a new source~$s_{i+1}$, a new sink~$t_{i+1}$, and four edges connecting the new source and sink with the old source and sink. Additionally, the latter two nodes are downgraded to ``normal'' nodes (nodes with a balance value of $0$) in~$G_{i+1}$ (see Figure~\ref{fig:step II}). By a careful choice of the new capacities and cost intervals we can ensure the following property: First, the SSP algorithm subsequently augments along all paths of the form
\[
s_{i+1} \to s_i \PATH[P] t_i \to t_{i+1} \COMMA
\]
where~$P$ is an $s_i$-$t_i$ path encountered by the SSP algorithm when run on the network~$G_i$. Then, it augments along all paths of the form
\[
  s_{i+1} \to t_i \PATH[\REVERSE{P}] s_i \to t_{i+1} \COMMA
\]
where~$P$ is again an $s_i$-$t_i$ path encountered by the SSP algorithm when run on the network~$G_i$. Hence, by adding two nodes and four edges we double the number of iterations the SSP algorithm requires. For this construction to work we have to double the maximum edge cost of our flow network. Hence, this construction can be repeated $k-1 \approx \log \phi$ times, yielding an additional factor of $2^{k-1} \approx \phi$ for the number of iterations required by the SSP algorithm.

\item In the third step we add a global source~$s$ and a global sink~$t$ to the flow network~$G_k$ constructed in the second step, and add four directed paths of length $M \approx \min \sSET{ n, \phi }$, where each contains $M$ new nodes and has exactly one node in common with $G_k$.
The first path will end in $s_k$, the second path will end in $t_k$, the third path will start in $s_k$, and the fourth path will start in $t_k$. We will also add an arc from~$s$ to every new node in the first two paths and an arc from every new node in the last two paths to~$t$ (see Figure~\ref{fig:step III}). We call the resulting flow network~$G$. By the right choice of the edge costs and capacities we will ensure that for each $s_k$-$t_k$ path~$P$ in~$G_k$ encountered by the SSP algorithm on~$G_k$ the SSP algorithm on~$G$ encounters~$M$ augmenting paths having~$P$ as a sub-path and~$M$ augmenting paths having~$\REVERSE{P}$ as a sub-path. In this way, we gain an additional factor of~$2M$ for the number of iterations of the SSP algorithm.
\end{enumerate}

In the following we say that the SSP algorithm \emph{encounters} a path~$P$ on a flow network~$G'$ if it augments along~$P$
when run on~$G'$.

\paragraph{Construction of~$G_1$.}
For the first step, consider two sets $U = \SET{ u_1, \ldots, u_n }$ and $W = \SET{ w_1, \ldots, w_n }$ of~$n$ nodes and an arbitrary set $E_{UW} \subseteq U \times W$ containing exactly $|E_{UW}| = m$ edges. The initial flow network~$G_1$ is defined as $G_1 = (V_1, E_1)$ for $V_1 = U \cup W \cup \SET{ s_1, t_1 }$ and
\[
  E_1 = (\SET{s_1} \times U) \cup E_{UW} \cup (W \times \SET{t_1}) \DOT
\]
The edges~$e$ from~$E_{UW}$ have capacity~$1$ and costs from the interval $I_e = [7, 9]$. 
The edges $(s_1,u_i), u_i \in U$ have a capacity equal to the out-degree of~$u_i$, 
the edges $(w_j,t_1), w_j \in W$ have a capacity equal to the in-degree of $w_j$ and both have costs from the interval $I_e = [0, 1$] (see Figure~\ref{fig:step I}).
(Remember that we use uniform distributions on the intervals $I_e$.)
\begin{fig}
  \GFX[width=0.5\textwidth]{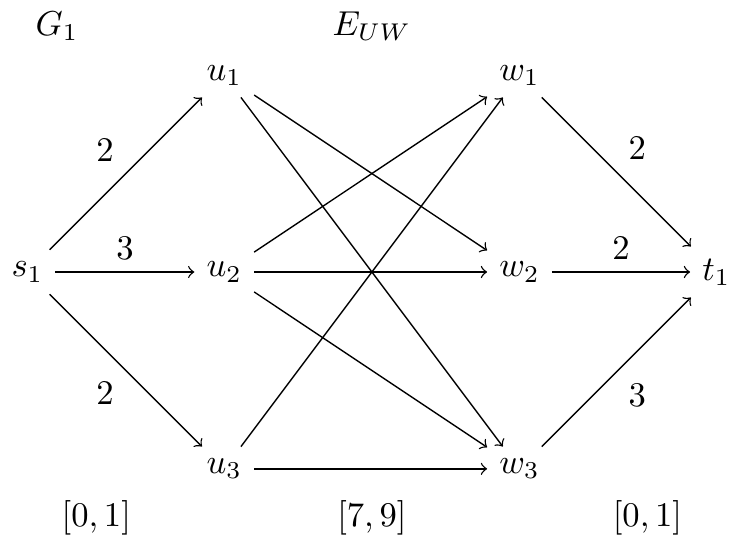}
 \caption[]
   {Example for $G_1$ with $n=3$ and $m = 7$ with capacities different from 1 shown next to the edges and the cost intervals shown below each edge set.} 
  \label{fig:step I}
\end{fig}

\begin{lemma}
\label{lemma:step I}
The SSP algorithm requires exactly~$m$ iterations on~$G_1$ to obtain a maximum $s_1$-$t_1$-flow. Furthermore all augmenting paths it
encounters have costs from the interval~$[7,11]$.
\end{lemma}
\begin{proof}
First we observe that the SSP algorithm augments only along paths that are of the form~$s_1\to u_i\to w_j\to t_1$ for some $u_i\in U$ and $w_j \in W$:
Consider an arbitrary augmenting path~$P$ the SSP algorithm encounters and assume for contradiction that~$P$ is not of this form.
Due to the structure of~$G_1$, the first two edges of~$P$ are of the form $(s_1,u_i)$ and $(u_i,w_j)$ for some $u_i\in U$ 
and $w_j \in W$. The choice of the capacities ensures that the edge~$(w_j,t_1)$ cannot be fully saturated if the edge~$(u_i,w_j)$ is not.
Hence, when the SSP algorithm augments along~$P$, the edge~$(w_j,t_1)$ is available in the residual network.
Since this edge is not used by the SSP algorithm, the sub-path~$w_j \PATH[P] t_1$ has smaller costs than the edge~$(w_j,t_1)$.
This means that the distance of~$w_j$ to the sink~$t_1$ in the current residual network is smaller than in the initial residual network
for the zero flow. This contradicts Lemma~\ref{lemma:distance monotonicity}.

Since every path the SSP algorithm encounters on~$G_1$ is of the form~$s_1\to u_i\to w_j\to t_1$, every such path consists of two edges with
costs from the interval~$[0,1]$ and one edge with costs from the interval~$[7,9]$. This implies that the total costs of any such path 
lie in the interval~$[7,11]$.

The choice of capacities ensures that on every augmenting path of the form~$s_1\to u_i\to w_j\to t_1$ the edge~$(u_i,w_j)$ is a bottleneck
and becomes saturated by the augmentation. As flow is never removed from this edge again, there is a one-to-one correspondence between
the paths the SSP algorithm encounters on~$G_1$ and the edges from~$E_{UW}$. This implies that the SSP algorithm encounters
exactly~$m$ paths on~$G_1$. 
\end{proof}

\paragraph{\boldmath Construction of~$G_{i+1}$ from~$G_i$.}
Now we describe the second step of our construction more formally. Given a flow network $G_i = (V_i, E_i)$ with a source~$s_i$ and a sink~$t_i$, we define $G_{i+1} = (V_{i+1}, E_{i+1})$, where $V_{i+1} = V_i \cup \SET{ s_{i+1}, t_{i+1} }$ and
\[
  E_{i+1} = E_{i} \cup (\SET{s_{i+1}} \times \SET{ s_i, t_i }) \cup (\SET{ s_i, t_i } \times \SET{t_{i+1}}) \DOT
\]
Let~$N_i = 2^{i-1}\cdot m$, which is the value of the maximum $s_i$-$t_i$ flow in~$G_i$. The new edges $e \in \SET{ (s_{i+1}, s_i), (t_i, t_{i+1}) }$ have capacity $u_e = N_i$ and costs from the interval $I_e = [0, 1]$. The new edges $e \in \SET{ (s_{i+1}, t_i), (s_i, t_{i+1}) }$ also have capacity $u_e = N_i$, but costs from the interval $I_e = [2^{i+3}-1, 2^{i+3}+1]$ (see Figure~\ref{fig:step II}).

\begin{fig}
  \GFX[width=0.5\textwidth]{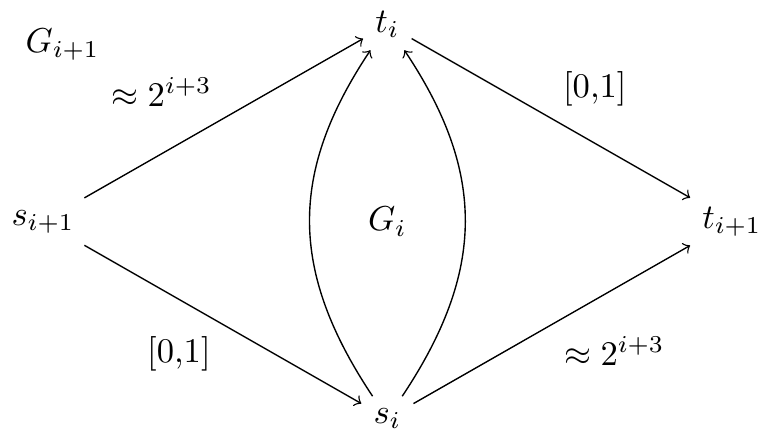}
 \caption[]
   {$G_{i+1}$ with $G_i$ as sub-graph with edge costs next to the edges.}
  \label{fig:step II}
\end{fig}

Next we analyze how many iterations the SSP algorithm requires to reach a maximum $s_{i+1}$-$t_{i+1}$ flow in~$G_{i+1}$ when run on the network~$G_{i+1}$.
Before we can start with this analysis, we prove the following property of the SSP algorithm.
\begin{lemma}
\label{lemma: shortest backward path} 
After augmenting flow via a cheapest $v$-$w$-path $P$ in a network without a cycle with negative total costs, $\stackrel{\leftarrow}{P}$ is a cheapest $w$-$v$-path. 
\end{lemma}
\begin{proof}
Since we augmented along $P$, all edges of $\stackrel{\leftarrow}{P}$ will be part of the residual network. $\stackrel{\leftarrow}{P}$ will therefore be a feasible $w$-$v$-path.
Assume that after augmenting along~$P$ there exists a $w$-$v$-path~$P'$
that is cheaper than~$\stackrel{\leftarrow}{P}$.
Let us take a look at the multi-set $X=P\cup P'$, which contains every arc~$e\in P\cap P'$ twice.
The total costs of this multi-set are negative because
\[
  c(P) + c(P') = -c(\stackrel{\leftarrow}{P}) + c(P') < 0
\] 
by the assumption that~$P'$ is cheaper than~$\stackrel{\leftarrow}{P}$.
Furthermore, for each node the number of incoming and outgoing arcs from~$X$ is the same.
This property is preserved if we delete all pairs of a forward arc~$e$ and the corresponding 
backward arc~$e^{-1}$ from~$X$, resulting in a multi-set~$X'\subseteq X$.
The total costs of the arcs in~$X'$ are negative because they equal the total costs of the arcs in~$X$.  

For every arc~$e\in X$ that did not have positive residual capacity before augmenting along~$P$, the arc~$e^{-1}$ must be part of $P$ and therefore be part of~$X$ as well.
This is due to the fact that only for arcs~$e$ with~$e^{-1} \in P$ the residual capacity increases when augmenting along $P$.
Since all such pairs of arcs are deleted, the set~$X'$ will only contain arcs that had a positive residual capacity
before augmenting along~$P$. Since each node has the same number of outgoing and
incoming arcs from~$X'$, we can partition~$X'$ into subsets, where the arcs in each subset
form a cycle. Since the total costs of all arcs are negative at
least one of these cycles has to have negative costs, which is a contradiction.
\end{proof}

Since during the execution of the SSP algorithm all residual networks have conservative costs on the arcs, Lemma~\ref{lemma: shortest backward path} always applies.
\begin{lemma}\label{lemma:CostsGi}
Let~$i\ge 1$.
All $s_i$-$t_i$-paths the SSP algorithm encounters when run on the network $G_i$ have costs from the interval~$[7, 2^{i+3}-5]$.
Furthermore the SSP algorithm encounters on  the network $G_{i+1}$ twice as many paths as on the network $G_{i}$.
\end{lemma}
\begin{proof}
We prove the first half of the lemma by induction over~$i$.
In accordance with Lemma~\ref{lemma:step I}, all paths the SSP algorithm encounters on~$G_1$ have costs from the interval~$[7,11] = [7,2^4-5]$.

Now assume that all paths the SSP algorithm encounters in~$G_{i}$, for some~$i\ge 1$, have costs from the interval~$[7, 2^{i+3}-5]$.
We distinguish between three different kinds of $s_{i+1}$-$t_{i+1}$-paths in $G_{i+1}$.
\begin{definition}
We classify the possible $s_{i+1}$-$t_{i+1}$-paths~$P$ in~$G_{i+1}$ as follows.
\begin{enumerate}\setlength{\itemsep}{0em}
  \item If $P = s_{i+1} \to s_i \PATH t_i \to t_{i+1}$, then~$P$ is called a \emph{type-1-path}.
  \item If $P = s_{i+1} \to s_i \to t_{i+1}$ or $P = s_{i+1} \to t_i \to t_{i+1}$, then~$P$ is called a \emph{type-2-path}.
  \item If $P = s_{i+1} \to t_i \PATH s_i \to t_{i+1}$, then~$P$ is called a \emph{type-3-path}.
\end{enumerate}
\end{definition}
For any type-2-path $P$ we have 
\[
  c(P) \in [0 + (2^{i+3} - 1), 1+ (2^{i+3} + 1)] = [2^{i+3} - 1, 2^{i+3} + 2]
  \subseteq [7, 2^{i+4}-5]\DOT
\]

Since due to Lemma~\ref{lemma:distance monotonicity} the distance from $t_{i}$
to $t_{i+1}$ does not decrease during the run of the SSP algorithm, the SSP
algorithm will only augment along a type-3-path~$P$ once the
edge~$(t_{i},t_{i+1})$ is saturated. Otherwise the
$t_{i}$-$t_{i+1}$-sub-path of $P$ could be replaced by the edge $(t_{i},t_{i+1})$
to create a cheaper path. Once the edge $(t_{i},t_{i+1})$ has been saturated,
the SSP algorithm cannot augment along type-1-paths anymore.
Therefore, the SSP algorithm will augment along all type-1-paths it encounters
before it augments along all type-3-paths it encounters.
 
Since during the time the SSP algorithm augments along type-1-paths no other
augmentations alter the part of the residual network corresponding to~$G_{i}$,
the corresponding sub-paths~$P'$ are paths in~$G_{i}$ that the
SSP algorithm encounters when run on the network~$G_i$.
Using the induction hypothesis, this yields that all type-1-paths the SSP algorithm encounters have costs
from the interval
\[
   [0 + 7 + 0,1 + \big(2^{i+3}-5\big) + 1] = [7, 2^{i+3}-3] \subseteq [7, 2^{i+4}-5] \DOT
\]
Since all of these type-1-paths have less costs than the two type-2-paths, the
SSP algorithm will augment along them as long as there still exists an
augmenting $s_{i}$-$t_{i}$-sub-path $P'$. Due to the choice of capacities this is
the case until both edges $(s_{i+1},s_{i})$ and $(t_{i},t_{i+1})$ are saturated.
Therefore, the SSP algorithm will not augment along any type-2-path.

When analyzing the costs of type-3-paths, we have to look at the $t_i$-$s_i$-sub-paths.
Let $\ell$ be the number of $s_i$-$t_i$-paths the SSP algorithm encounters when
run on the network $G_i$ and let $P_1,P_2,\ldots,P_{\ell}$ be the corresponding paths
in the same order, in which they were encountered.
Then Lemma~\ref{lemma: shortest backward path} yields that for any $j \in
\{1,\ldots,{\ell}\}$ after augmenting along the paths $P_1,P_2,\ldots,P_j$ the cheapest
$t_i$-$s_i$-path in the residual network is $\stackrel{\leftarrow}{P_j}$. Property~\ref{property:different path lengths} yields that it is the only cheapest path. Also
the residual network we obtain, if we then augment via
$\stackrel{\leftarrow}{P_j}$ is equal to the residual network obtained, when
only augmenting along the paths $P_1,P_2,\ldots,P_{j-1}$. Starting with $j={\ell}$ this
yields that the $t_i$-$s_i$-sub-paths corresponding to the type-3-paths the SSP
algorithm encounters are equal to
$\stackrel{\leftarrow}{P_{\ell}},\ldots,\stackrel{\leftarrow}{P_1}$.
By induction the cost of each such path $P_j$ lies in $[7, 2^{i+3}-5]$. This
yields that every type-3-path the SSP algorithm encounters has costs from the interval
\begin{align*}
& [(2^{i+3} - 1) - (2^{i+3}-5)  + (2^{i+3} - 1), (2^{i+3} + 1)- 7 + (2^{i+3} + 1)] \\
 = \,\, & [2^{i+3} + 3, 2^{i+4} - 5] \subseteq [7, 2^{i+4}-5] \DOT
\end{align*}

The previous argument also shows that the SSP algorithm encounters on~$G_{i+1}$ twice as many paths as on~$G_i$
because it encounters~$\ell$ type-1-paths, no type-2-path, and~$\ell$ type-3-paths, where~$\ell$ denotes the
number of paths the SSP algorithm encounters on~$G_i$. 
\end{proof}

Since the SSP algorithm augments along $m$ paths when run on the network $G_1$, it will augment along $2^{i-1}\cdot m$ paths when run on the network $G_i$.
Note, that at the end of the SSP algorithm, when run on $G_i$ for $i > 1$, only the $4$ arcs incident to $s_i$ and $t_i$ carry flow.

\paragraph{Construction of~$G$ from~$G_k$.}
Let $N_k = 2^{k-1}\cdot m$, which is the value of a maximum $s_k$-$t_k$ flow in $G_k$. We will now use $G_k$ to define $G=(V,E)$ as follows
(see also Figure~\ref{fig:step III}).
\begin{itemize}
\setlength{\itemsep}{0cm}
\item $V := V_k \cup A \cup B \cup C \cup D \cup \{s,t\}$, with $A := \{a_1,a_2,\dots,a_M\}$, $B := \{b_1,b_2,\dots,b_M\}$, $C := \{c_1,c_2,\dots,c_M\}$, and $D := \{d_1,d_2,\dots,d_M\}$. $E:= E_k \cup E_a \cup E_b \cup E_c \cup E_d$.\newline
\item $E_a$ contains the edges $(a_i,a_{i-1})$, $i \in \{2,\dots,M\}$, with cost interval $[2^{k+5}-1,2^{k+5}]$ and infinite capacity,
$(s,a_i)$, $i \in \{1,\dots,M\}$, with cost interval $[0,1]$ and capacity $N_k$, and $(a_1,s_k)$ with cost interval $[2^{k+4}-1,2^{k+4}]$ and infinite capacity.\newline
\item $E_b$ contains the edges $(b_i,b_{i-1})$, $i \in \{2,\dots,M\}$, with cost interval $[2^{k+5}-1,2^{k+5}]$ and infinite capacity,
$(s,b_i)$, $i \in \{1,\dots,M\}$, with cost interval $[0,1]$ and capacity $N_k$, and $(b_1,t_k)$ with cost interval $[2^{k+5}-1,2^{k+5}]$ and infinite capacity.\newline
\item $E_c$ contains the edges $(c_{i-1},c_i)$, $i \in \{2,\dots,M\}$, with cost interval $[2^{k+5}-1,2^{k+5}]$ and infinite capacity,
$(c_i,t)$, $i \in \{1,\dots,M\}$, with cost interval $[0,1]$ and capacity $N_k$, and $(s_k,c_1)$ with cost interval $[2^{k+5}-1,2^{k+5}]$ and infinite capacity.\newline
\item $E_d$ contains the edges $(d_{i-1},d_i)$, $i \in \{2,\dots,M\}$, with cost interval $[2^{k+5}-1,2^{k+5}]$ and infinite capacity,
$(d_i,t)$, $i \in \{1,\dots,m\}$, with cost interval $[0,1]$ and capacity $N_k$, and $(t_k,d_1)$ with cost interval $[2^{k+4}-1,2^{k+4}]$ and infinite capacity.
\end{itemize}

\begin{fig}
  \GFX[width=0.95\textwidth]{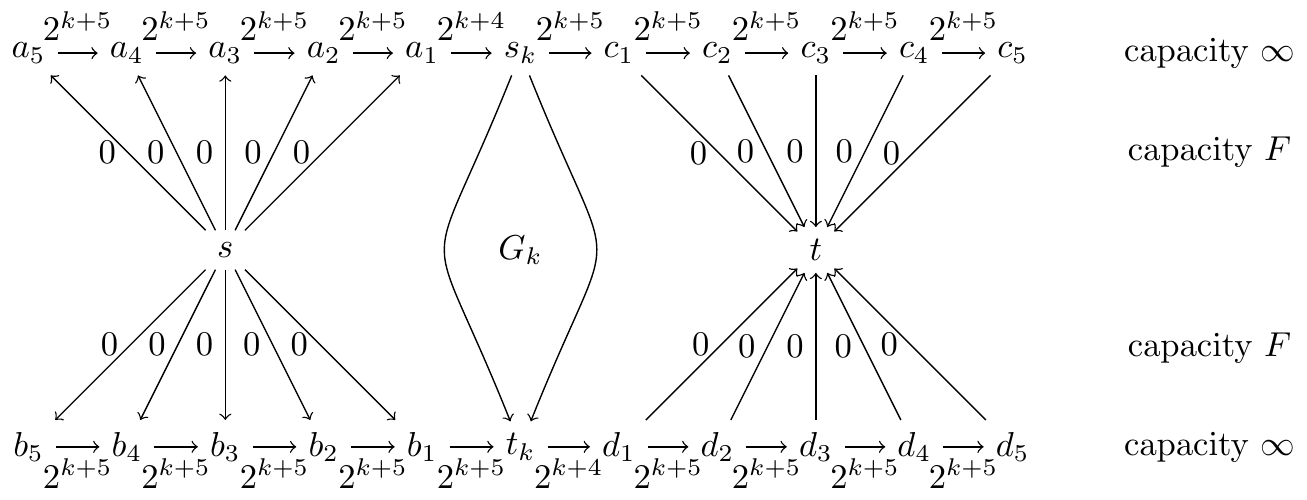}
  \caption{$G$ with $G_k$ as sub-graph with approximate edge costs on the edges. A value $c$ below an edge $e$ means the the cost of $e$ is drawn uniformly at random from the interval $[c-1,c]$.}
  \label{fig:step III}
\end{fig}

\begin{theorem}
The SSP algorithm encounters $m \cdot 2^{k-1} \cdot 2M$ paths on the network~$G$.
\end{theorem}

\begin{proof}
We categorize the different $s$-$t$-paths the SSP algorithm encounters on $G$ by the node after $s$ and the node before~$t$.
Each such $s$-$t$-path can be described as an $\{a_i,c_j\}$-, $\{a_i,d_j\}$-, $\{b_i,c_j\}$-, or $\{b_i,d_j\}$-path for some $i,j \in \{1,\dots,M\}$.

All $s_k$-$t_k$-paths encountered by the SSP algorithm, when run on $G_k$, have costs from the interval~$[7, 2^{k+3}-5]$ in accordance with Lemma~\ref{lemma:CostsGi}.
For any~$i\in\{1,\ldots,m\}$, the costs of the $s$-$a_i$-$s_k$-path and the $t_k$-$d_i$-$t$-path lie
in $[\alpha_i,\alpha_i+(i+1)]$ with $\alpha_i = 2^{k+5}i - 2^{k+4} - i$ and the costs of the $s$-$b_i$-$t_k$-path and the
$s_k$-$c_i$-$t$-path lie in $[\beta_i,\beta_i+(i+1)]$ with $\beta_i = 2^{k+5}i-i$.
Furthermore $i<M+1< 2^{k+3}$.

Therefore, the SSP algorithm will
only augment along $\{a_i,c_j\}$-paths if no $\{a_i,d_j\}$-paths are available.
Also, any $\{a_i,d_i\}$-path is shorter than any $\{b_i,c_i\}$-path and any $\{b_i,c_i\}$-path is shorter than any
$\{a_{i+1},d_{i+1}\}$-path. Finally, any $\{b_i,c_j\}$-path is shorter than any $\{a_{i+1},c_j\}$-path or $\{b_i,d_{j+1}\}$-path.
Therefore, the SSP algorithm will start with augmenting along $\{a_1,d_1\}$-paths. After augmenting along $\{a_i,d_i\}$-paths it will augment
along $\{b_i,c_i\}$-paths and after augmenting along $\{b_i,c_i\}$-paths it will augment along $\{a_{i+1},d_{i+1}\}$-paths.
Due to the choice of the capacities we can see that once the SSP algorithm starts augmenting along an $\{a_i,d_i\}$-path it keeps augmenting
along $\{a_i,d_i\}$-paths until there is no $s_k$-$t_k$-path in the residual network that lies completely in the sub-network corresponding
to $G_k$. Also, once the SSP algorithm starts augmenting along an $\{b_i,c_i\}$-path it keeps augmenting along $\{b_i,c_i\}$-paths until
there is no $t_k$-$s_k$-path in the residual network that lies completely in the sub-network corresponding  to $G_k$.
After the SSP algorithm augmented along the last $\{a_i,d_i\}$-path the residual network in the sub-network corresponding to $G_k$ is equal
to the residual network of a maximum flow in $G_k$. After the SSP algorithm augmented along the last $\{b_i,c_i\}$-path the residual network
in the sub-network corresponding to $G_k$ is equal to $G_k$. We can see that the SSP algorithm augments along an $\{a_i,d_i\}$-path for every path $P$ it encounters on $G_k$ and along an
$\{b_i,c_i\}$-path for the backwards path $\stackrel{\leftarrow}{P}$ of every path $P$ it encounters on $G_k$.
Therefore, the SSP-algorithm will augment $M$ times along paths corresponding to the paths it encounters on $G_k$ and $M$ times along
paths corresponding to the backward paths of these paths and therefore augment along $2M$ times as many paths in
$G$ as in $G_k$.
\end{proof}

To show that $G$ contains $2n+2k+2+4M$ nodes and $m+2n+4k-4+8M$ edges, we
observe that $G_1$ has $2n+2$ nodes and $m+2n$ edges, the $k-1$ iterations to
create $G_k$ add a total of $2k-2$ nodes and $4k-4$ edges and the construction
of $G$ from $G_k$ adds $4M+2$ nodes and $8M$ edges. This gives a total of
$2n+2+2k-2+4M+2 = 2n+2k+2+4M$ nodes and $m+2n+4k-4+8M$ edges. Since $k, M =
O(n)$ and $m\geq n$, $G$ has $O(n)$ nodes and $O(m)$ edges and forces the SSP
algorithm to encounter $m \cdot 2^{k-1} \cdot 2M = \Omega(m \phi M)=\Omega(\phi
\cdot m \cdot \min(\phi, n))$ paths on $G$. For $\phi = \Omega(n)$ this lower
bound shows that the upper bound of $O(mn\phi)$ augmentation steps in
Theorem~\ref{maintheorem} is tight.

\section{Smoothed Analysis of the Simplex Algorithm}
\label{sec:simplex}

In this section we describe a surprising connection between our result about the SSP algorithm and the smoothed analysis of the
simplex algorithm. Spielman and Teng's original smoothed analysis~\cite{DBLP:journals/jacm/SpielmanT04} as well as 
Vershynin's~\cite{DBLP:journals/siamcomp/Vershynin09} improved analysis are based on the shadow vertex method. To describe this
pivot rule, let us consider a linear program with an objective function~$z^Tx$ and a set of constraints~$Ax\le b$. Let us
assume that a non-optimal initial
vertex~$x_0$ of the polytope~$P$ of feasible solutions is given. The shadow vertex method computes an objective function~$u^Tx$ that is
optimized by~$x_0$. Then it projects the polytope~$P$ onto the 2-dimensional plane that is spanned by the vectors~$z$ and~$u$. If we assume
for the sake of simplicity that~$P$ is bounded, then the resulting projection is a polygon~$Q$.

The crucial properties of the polygon~$Q$ are as follows: both the projection of~$x_0$ and the projection of the optimal solution~$x^*$ are vertices of~$Q$, and every edge of~$Q$ corresponds to an edge of~$P$. The shadow vertex method follows the edges of~$Q$ from the projection
of~$x_0$ to the projection of~$x^*$. The aforementioned properties guarantee that this corresponds to a feasible walk on the polytope~$P$. 

To relate the shadow vertex method and the SSP algorithm, we consider the canonical linear program for the maximum-flow problem with one
source and one sink. In this linear program, there is a variable for each edge corresponding to the flow on that edge.
The objective function, which is to be
maximized, adds the flow on all outgoing edges of the source and subtracts the flow on all incoming edges of the source. There are constraints for each edge 
ensuring that the flow is non-negative and not larger than the capacity, and there is a constraint for each node except the source and the sink ensuring
Kirchhoff's law.

The empty flow~$x_0$ is a vertex of the polytope of feasible solutions. In particular, it is a feasible solution with minimum costs. Hence, letting~$u$ be
the vector of edge costs is a valid choice in the shadow vertex method. For this choice every feasible flow~$f$ is projected to the pair~$(|f|, c(f))$. 
Theorem~\ref{thm:AllFlowsOpt} guarantees that the cost function depicted in Figure~\ref{fig:example cost function} forms the lower envelope of the  
polygon that results from projecting the set of feasible flows. There are two possibilities for the shadow vertex method for the first step: it can choose
to follow either the upper or the lower envelope of this polygon. If it decides for the lower envelope, then it will encounter exactly the same sequence of flows
as the SSP algorithm.

This means that Theorem~\ref{maintheorem} can also be interpreted as a statement about the shadow vertex method applied to the maximum-flow linear
program. It says that for this particular class of linear programs, the shadow vertex method has expected polynomial running time even if the linear program is chosen
by an adversary. It suffices to perturb the costs, which determine the projection used in the shadow vertex method. Hence, if the projection
is chosen at random, the shadow vertex method is a randomized simplex algorithm with polynomial expected running time
for any flow linear program. 

In general, we believe that it is an interesting question to study whether the strong assumption in Spielman and Teng's~\cite{DBLP:journals/jacm/SpielmanT04} and 
Vershynin's~\cite{DBLP:journals/siamcomp/Vershynin09} smoothed analysis that all coefficients in the constraints are perturbed is necessary. In particular, we find
it an interesting open question to characterize for which class of linear programs it suffices to perturb only the coefficients in the objective
function or just the projection in the shadow vertex method to obtain polynomial smoothed running time.

Two of us have studied a related question~\cite{BrunschR13}. 
We have proved that the shadow vertex method can be used
to find short paths between given vertices of a polyhedron. Here, short
means that the path length is~$O(\frac{mn^2}{\delta^2})$, where~$n$ denotes the number of
variables, $m$ denotes the number of constraints, and~$\delta$ is a parameter 
that measures the flatness of the vertices of the polyhedron.
This result is proven by a significant extension of the analysis presented in this article.

\end{document}